\documentclass[12pt, draftclsnofoot, onecolumn]{IEEEtran}

\usepackage{amsmath}
\usepackage{amsthm}
\usepackage{amsbsy}
\usepackage{amssymb}
\usepackage{bm}
\usepackage[noadjust]{cite}
\usepackage{graphicx}
\usepackage{float}
\usepackage{epsfig}
\usepackage{microtype}
\usepackage{mathtools}
\usepackage{latexsym}
\usepackage{wasysym}
\usepackage{marvosym}
\usepackage[subtle]{savetrees}
\usepackage{placeins}
\usepackage{color}
\usepackage{rotating}
\usepackage[usenames,dvipsnames]{xcolor}
\usepackage{caption}
\usepackage{subcaption}
\usepackage{epstopdf}
\usepackage{cancel}
\usepackage{tikz}
\usetikzlibrary{fit,positioning,arrows.meta}
\usepackage{tkz-euclide}
\usepackage{algorithm}
\usepackage{algpseudocode}
\usepackage{overpic}
\usepackage{dashrule}
\usepackage{subdepth}
\usepackage{overpic}
\usepackage{stackengine}
\usepackage{accents}
\usepackage{stmaryrd}
\usepackage{acro}
\NewDocumentCommand{\acro}{m o m o}
{%
	\IfValueTF{#2}{%
		\IfValueTF{#4}{%
			\DeclareAcronym{#1}{short={#2},long={#3},#4}
		}{%
			\DeclareAcronym{#1}{short={#2},long={#3}}
		}
	}{%
		\IfValueTF{#4}{%
			\DeclareAcronym{#1}{short={#1},long={#3},#4}
		}{%
			\DeclareAcronym{#1}{short={#1},long={#3}}
		}
	}
}
 
\newtheorem{definition}{Definition}
\newtheorem{proposition}{Proposition}
\newtheorem{lemma}{Lemma}
\theoremstyle{remark}
\newtheorem{corollary}{Corollary}
\theoremstyle{remark}
\newtheorem{remark}{Remark} %[section]
\theoremstyle{example}
\theoremstyle{assumption}
\DeclareMathAlphabet{\mathbit}{OML}{cmr}{bx}{it}
\DeclareMathAlphabet{\mathsf}{OT1}{cmss}{m}{n}
\DeclareMathAlphabet{\mathTXf}{OT1}{cmss}{bx}{it}

\DeclareMathOperator*{\argmax}{argmax}

\newcommand{\norm}[1]{\lVert{#1}\rVert}

\DeclareAcronym{UE}{
	short = UE,
	short-plural-form = UEs,
	long = user equipment,
	long-plural-form = user equipments,
	foreign-plural = {}
}\DeclareAcronym{BS}{
	short = BS,
	short-plural-form = BSs,
	long = base station,
	long-plural-form = base stations,
	foreign-plural = {}
}
\acro{CSI}{channel state information}
\acro{CSIT}{channel state information at the transmitter}
\acro{mMIMO}{massive MIMO}
\acro{MIMO}{multiple-input multiple-output}
\acro{TDD}{time division duplex}
\acro{FDD}{\em frequency division duplex}
\acro{GoB}{grid-of-beams}
\acro{AP}{antenna port}
\acro{NR}{New Radio}
\acro{D2D}{device-to-device}
\acro{LS}{least square}
\acro{SINR}{signal-to-interference-and-noise ratio}
\acro{SNR}{signal-to-noise ratio}
\acro{ZF}{Zero-Forcing}
\acro{MMSE}{minimum mean square error}
\acro{NMSE}{normalized mean square error}
\acro{LMMSE}{linear minimum mean square error}
\acro{RF}{radio frequency}
\acro{SE}{spectral efficiency}
\acro{OFDM}{orthogonal frequency-division multiplexing}
\acro{CS}{compressive sensing}
\acro{mmWave}{millimeter wave}
\acro{SVD}{singular value decomposition}
\acro{EVD}{eigenvalue decomposition}
\acro{PMI}{precoding matrix indicator}
\acro{MSE}{mean square error}
\acro{BD}{block diagonalization}
\acro{AO}{alternating optimization}
\acro{UL}{uplink}
\acro{DL}{downlink}
\acro{DFT}{discrete Fourier transform}
\acro{ML}{machine learning}
\acro{LI}{linearly independent}
\acro{CMD}{correlation matrix distance}
\acro{GCMD}{generalized correlation matrix distance}
\DeclareAcronym{RB}{
	short = RB,
	short-plural-form = RBs,
	long = resource block,
	long-plural-form = resource blocks,
	foreign-plural = {}
}
\DeclareAcronym{RS}{
	short = RS,
	short-plural-form = RSs,
	long = reference signal,
	long-plural-form = reference signals,
	foreign-plural = {}
}
\DeclareAcronym{AoD}{
	short = AoD, 
	short-plural-form = AoD, 
	long = angle-of-departure, 
	long-plural-form = angles-of-departure,
	foreign-plural = {}
}
\DeclareAcronym{AoA}{
	short = AoA, 
	short-plural-form = AoA, 
	long = angle-of-arrival, 
	long-plural-form = angles-of-arrival,
	foreign-plural = {}
}

\IEEEoverridecommandlockouts
\definecolor{orange}{rgb}{0.8627,0.4314,0.1961}
\definecolor{grey}{rgb}{0.2196,0.2471,0.3176}
\definecolor{blue}{rgb}{0.0275,0.4431,0.5294}
\definecolor{green}{rgb}{0,0.6078,0.4392}
\definecolor{redd}{rgb}{0.887, 0.07, 0.27}
\definecolor{orange}{rgb}{0.8627,0.4314,0.1961}
\definecolor{grey}{rgb}{0.2196,0.2471,0.3176}
\definecolor{blue}{rgb}{0.0275,0.4431,0.5294}
\definecolor{green}{rgb}{0,0.6078,0.4392}
\definecolor{EurecomGrey}{RGB}{188 188 188}
\definecolor{EurecomBlue}{RGB}{0 162 225}
\definecolor{EurecomBlueDark}{RGB}{0 73 102}

\begin{document}

\bstctlcite{IEEEexample:BSTcontrol}
\title{User Coordination for Fast Beam Training \\ in FDD Multi-User Massive MIMO}
\author{
     \IEEEauthorblockN{
     	Flavio Maschietti, \emph{Member, IEEE},
     	G\'abor Fodor, \emph{Senior Member, IEEE},\\
     	David Gesbert, \emph{Fellow, IEEE}, 
     	Paul de Kerret, \emph{Member, IEEE}\\
	}
	\thanks{
		This work was in part supported by the ERC under the European Union's Horizon 2020
		research and innovation program (Agreement no. 670896 PERFUME).
		F. Maschietti, D. Gesbert and P. de Kerret are with EURECOM, Sophia-Antipolis, France 
		(e-mail: \{flavio.maschietti, david.gesbert, paul.dekerret\}@eurecom.fr).
		G. Fodor is with Ericsson Research, Kista, Sweden and KTH, Stockholm, Sweden (e-mail: gabor.fodor@ericsson.com).
		Part of this work has been carried out while F. Maschietti was visiting the radio
		department of Ericsson Research in Kista, Sweden, and
		has been published in the proceedings of the IEEE ISWCS $2019$~\cite{Maschietti2019c}. 
		The authors would like to thank A. Bazco-Nogueras at EURECOM and G. Klang at Ericsson Research 
		for their support.
	}
}
\IEEEtitleabstractindextext{
	
	\vspace{-0.77cm}
	\begin{abstract}

		Massive multiple-input multiple-output (mMIMO) communications are one of the enabling technologies 
		of $5$G and beyond networks. While prior work indicates that mMIMO networks employing time division
		duplexing have a significant capacity growth potential, deploying mMIMO in frequency division 
		duplexing (FDD) networks remains problematic. The two main difficulties in FDD networks are the
		scalability of the downlink reference signals and the overhead associated with the required uplink
		feedback for channel state information (CSI) acquisition. To address these 
		difficulties, most existing methods utilize assumptions on the radio environment such as channel sparsity 
		or angular reciprocity. In this work, we propose a novel cooperative method for a scalable and low-overhead 
		approach to FDD mMIMO under the so-called grid-of-beams architecture. The key idea behind our scheme lies in 
		the exploitation of the near-common signal propagation paths that are often found across several mobile users 
		located in nearby regions, through a coordination mechanism. In doing so, we leverage the recently specified
		device-to-device communications capability in 5G networks. Specifically, we design beam selection
		algorithms capable of striking a balance between CSI acquisition overhead and multi-user interference
		mitigation. The selection exploits statistical information, through so-called covariance shaping.
		Simulation results demonstrate the effectiveness of the proposed algorithms, which prove particularly
		well-suited to rapidly-varying channels with short coherence time.
		
	\end{abstract}
	
}

\maketitle

\IEEEdisplaynontitleabstractindextext
		
\section{Introduction} \label{sec:Intro}

	\Ac{mMIMO} is expected to enable higher performance in $5$G and beyond networks through 
	increased data rate, more reliable and power-efficient radio links and reduced multi-user 
	interference~\cite{Larsson2014}. The \ac{mMIMO} concept originated in a \ac{TDD} setting,
	where exploiting the \emph{channel reciprocity} through low-overhead orthogonal \ac{UL} sounding led 
	to the design of near-optimal linear precoders~\cite{Larsson2014}. In contrast, \ac{DL} \acp{RS} and
	subsequent \ac{UL} feedback are required to estimate the \ac{DL} channels in \ac{FDD} mode, which makes 
	it considerably more challenging. In general, there exists a one-to-one correspondence between \acp{RS}
	and antenna elements. Therefore, training and feedback overhead are often associated with
	\emph{unfeasibility} in the \ac{FDD} \ac{mMIMO} regime, where too few resource elements would be in
	principle left for data transmission~\cite{Choi2014}.
	
	Nevertheless, operating in \ac{FDD} remains appealing to mobile operators for several reasons, including
	\emph{i)} most radio bands below $6$ GHz are paired \ac{FDD} bands, \emph{ii)} the \acp{BS} have higher
	transmit power available for the \acp{RS} than the \acp{UE}, \emph{iii)} overall deployment, maintenance
	and operation costs are reduced as fewer \acp{BS} are required in \ac{FDD} networks~\cite{Choi2014}.
	
	\subsection{Related Work}
	
		Several papers have proposed methods to cope with the overhead issue in \ac{FDD}
		\ac{mMIMO}. In this section, we provide a short overview of such existing works, which can be
		divided in four categories: \emph{i)} second-order statistics-based approaches, 
		\emph{ii)} \ac{CS}-based approaches, \emph{iii)} channel extrapolation-based approaches, and 
		\emph{iv)} \ac{GoB}-based approaches. 
		
		Among the approaches based on second-order statistics, the work such as~\cite{Yin2013, Adhikary2013}
		demonstrated that -- under \emph{strictly spatially-orthogonal} low-rank channel covariances -- it is
		possible to discriminate across interfering \acp{UE} with even correlated non-orthogonal pilot
		sequences, thus reducing the training overhead and the so-called pilot contamination~\cite{Larsson2014}.
		In~\cite{Adhikary2013}, the low-rankness allows beamforming with low-dimensional \ac{CSI} at the
		\ac{BS}. However, such condition is seldom experienced in practical scenarios where the channel
		components are spread over the angular domain and unlikely result in non-overlapping channel
		eigenspaces~\cite{Gao2015}. In general, the radio environment has an important role in coloring the
		channel covariance. The more recent work~\cite{Khalilsarai2019} has introduced a covariance-based
		precoding method to \emph{artificially} forge low-dimensional effective channels, \emph{independently}
		from the covariance structure. In the MIMO literature, such precoding methods have been known under 
		the term \emph{covariance shaping}~\cite{Newinger2015, Moghadam2017, Mursia2018}.

		On a different note, \ac{CS} techniques for estimating high-dimensional sparse channels with only 
		a few measurements have been known for decades~\cite{Bajwa2010} and have been applied to \ac{FDD} 
		\ac{mMIMO} as well~\cite{Rao2014, Gao2016, Shen2016, Dai2018}. The training overhead reduction in 
		all these works relies on the existence of an intrinsic sparse representation of the radio channels,  
		although this is not always found in networks operating at sub-$6$ GHz bands~\cite{Gao2015,
		Martinez2016}. Alternative \ac{CS}-based methods such as~\cite{Ding2018} capitalize on the 
		\emph{angular reciprocity} between the \ac{DL} and the \ac{UL} channels. In such approaches, 
		the spatial spectrum is estimated from \ac{UL} sounding and used to design the \ac{mMIMO} precoder, 
		under the reasonable assumption that the dominant \acp{AoD} are almost invariant over the spectrum 
		range separating the \ac{DL} and the \ac{UL} channels. Similar angle-based methods can be found 
		in~\cite{Luo2017, Zhang2018, Shen2018}. Nevertheless, the presumed angular correlation can decrease 
		in some practical scenarios, due to e.g. carrier aggregation, and lead to performance degradation.
		
		The intuition behind the channel extrapolation-based approaches is to infer the \ac{DL} \ac{CSI}
		from \ac{UL} pilot estimates~\cite{Rottenberg2020}. Therefore, the complete elimination of the 
		\ac{DL} training overhead is achieved with those approaches, as in a genuine TDD setting.
		The pioneering work in~\cite{Vasisht2016} develops a transform that can infer \ac{DL} parameters 
		such as path distance and gain from \ac{UL} channels measured at the \ac{BS}. In~\cite{Yang2018}, 
		the authors propose to use the \emph{super-resolution theory} for achieving the \ac{DL} channel 
		extrapolation. \Ac{ML}-based techniques have been proposed for channel extrapolation as well.
		In~\cite{Arnold2019}, the \ac{DL} \ac{CSI} is predicted from adjacent \ac{UL} bands through a
		deep neural network, while in~\cite{Yang2019} a complex-valued neural network is trained
		to approximate a deterministic \ac{UL}-\ac{DL} mapping function. Further efforts are made
		in~\cite{Choi2020}, where a \ac{DL} \ac{CSI} extrapolation technique using a neural network with 
		simplified input and output is proposed to reduce the learning time.
			
		The \ac{GoB} approach has recently raised much interest especially within the $3$GPP fora, due to its
		\emph{practical implementability}~\cite{3GPP2018c}. The idea is again to translate high-dimensional
		channels into low-dimensional representations. Indeed, according to this concept, reduced channel
		representations are obtained through a spatial transformation based on fixed transmit-receive 
		beams~\cite{Dahlman2018, Kim2013, 3GPP2018c}. Thus, the \acp{UE} see low-dimensional effective channels 
		which incorporate the beamforming vectors relative to the beams. In this case, there exists a 
		one-to-one correspondence between \acp{RS} and beams in the \ac{GoB} codebook~\cite{Kim2013, Dahlman2018, 
		3GPP2018c}. Therefore, estimating such effective channels reduces the training overhead, as it
		becomes proportional to the codebook size and independent from the number of antenna elements. 
		On the upside, the \ac{GoB} approach allows a low-dimensional representation even when no sparse
		representation of the \ac{DL} channels exists. However, the reduction in training (and feedback)
		overhead again entails a \emph{possibly severe} performance degradation~\cite{Flordelis2018}. 
		This is because the \ac{mMIMO} data precoder is optimized for reduced channel representations 
		which might not capture the prominent characteristics of the actual radio channels.
	
		In order to minimize such losses, an alternative consists in designing the \ac{GoB} with a larger
		number of beams, and then training a small subset of them which contains the most relevant channel
		components~\cite{Zirwas2016, Xiong2017, Flordelis2018}. The number of such components depends on the propagation
		environment, which is in general beyond the designer's control. An interesting twist to the story
		arises when multiple antennas are considered at the \ac{UE} side \emph{as well}, as we did 
		in~\cite{Maschietti2019c}. In fact, in that case, an extra degree of freedom is obtained by letting 
		the \acp{UE} steer energy into suitable spatial regions. In particular, statistical beamforming at the
		\ac{UE} side \emph{(covariance shaping)} can be used to excite desirable channel subspaces. 
		In conventional cases, UE-based beamforming focuses on regions where strong paths are located. 
		While this approach makes much sense in a single-user scenario, it fails to exploit all degrees of
		freedom in a \emph{multi-user} setting. In fact, ~\cite{Maschietti2019c} set forth the idea that beam
		selection at the UE side could be designed differently, with the aim to reduce the number of relevant
		components to estimate (so as to reduce the training overhead). In general, the decision on which beams
		to activate at both the \ac{BS} and \ac{UE} sides is not a trivial one, as several factors participate
		in the sum-rate optimization problem, including \emph{i)} the beamforming gain, 
		\emph{ii)} the multi-user interference and \emph{iii)} the training overhead. Furthermore, the beam
		selection is ideally a \emph{joint} decision problem across the \acp{UE} and the \ac{BS} and a
		coordination problem ensues. It is precisely the exploration of the novel \emph{three-way trade-off}
		arising from the factors \emph{i)}, \emph{ii)}, \emph{iii)} and underlying coordination mechanisms 
		that form the core ideas of the paper.
		
	\subsection{Contributions} %% MUST BE CLEAR OF THE ENGINEERING APPROACH/POV
		
		In this work, we propose a coordination mechanism between the \acp{UE} to facilitate statistical 
		beam selection for \ac{FDD} \ac{mMIMO} performance optimization under the \ac{GoB} design. 
		We consider multi-beam selection at the \ac{UE} side with multi-stream \ac{mMIMO} transmission based 
		on the \ac{BD} precoding. This paper shows that beam selection in \ac{FDD} \ac{mMIMO} involves an
		interesting \emph{trade-off} between \emph{i)} selecting the beams that capture the largest channel
		gains for each \ac{UE}, i.e. the most relevant channel components, and \emph{ii)} selecting the beams
		that might capture somewhat weaker paths but are \emph{common} to multiple \acp{UE}, so as to reduce 
		the training overhead. The essence of such trade-offs is captured in Fig.~\ref{fig:scen}, where \ac{UE}
		$2$ can capitalize on its weaker paths to reduce the number of activated beams at the \ac{BS} side.
		Nevertheless, focusing only on beams that are common to multiple \acp{UE} can lead to decreased
		\emph{spatial separability} among them and reduce the gains. To this end, we introduce the so-called
		\ac{GCMD}, a metric to evaluate the impact of covariance shaping on the average spatial separation 
		of the \acp{UE}.
		
		In order to design the long-term \ac{GoB} beamformers, we propose a suite of (decentralized) 
		coordinated beam selection algorithms exploring various \emph{complexity-performance} trade-offs. 
		In practice, the coordination between the \acp{UE} is enforced through a message exchange protocol
		exploiting low-rate \ac{D2D} communications. In this respect, we leverage from the forthcoming $3$GPP
		Release $16$, which is expected to support point-to-point \emph{side-links} that facilitate
		cooperative communications among neighboring \acp{UE} with low resource consumption~\cite{Fodor2016,
		3GPP2018d}. \enlargethispage{-\baselineskip}
		% The \ac{NR} \emph{side-link} is thus a cornerstone for the proposed schemes.
		% The final goal of the proposed algorithms is to strike a balance between \emph{i)} harvesting the 
		% largest possible effective channel gain, \emph{ii)} avoiding multi-user interference, and 
		% \emph{iii)} minimizing the training overhead.
		
		Our numerical results -- under the Winner II channel model -- show that beam-domain coordination improves
	 	the throughput performance of \ac{GoB}-based \ac{FDD} \ac{mMIMO} as compared with uncoordinated beam selection. 
	 	In particular, the highest gains over uncoordinated \acs{SNR}-based beam selection are experienced for 
	 	\emph{rapidly-varying} channels, such as the vehicular or the pedestrian ones. For such channels, fast beam 
	 	training is essential to cope with the short channel coherence time below $20$ ms.
	 	As a design lesson, we show that shaping the covariance matrices so as to favor the \emph{spatial separability} 
	 	(or \emph{orthogonality}) among the \acp{UE} -- such as in~\cite{Newinger2015, Mursia2018} -- is detrimental 
	 	in such fast channels. This is because such approach neglects the optimization of the pre-log factor related 
	 	to the training overhead, which has a substantial impact on the effective network throughput.
						
	\subsection{Notation}
		
		We use the following notation throughout the paper: bold lowercase letters are reserved for vectors,
		while bold uppercase letters for matrices. 
		The conjugate operator is denoted with $\text{conj}\left(\cdot\right)$, the transpose operator is
		denoted with $\left(\cdot\right)^{\text{T}}$, while the Hermitian transpose operator is 
		$\left(\cdot\right)^{\text{H}}$. The expectation operator over the random variable $X$ is denoted with 
		$\mathbb{E}_X\left[\cdot\right]$. $\text{Tr}\left(\cdot\right)$ denotes the trace operator; 
		$\det\left(\cdot\right)$ denotes the determinant operator, while $\text{rank}\left(\cdot\right)$
		denotes the rank operator. We denote with $\text{col}\left(\cdot\right)$ (resp. 
		$\text{row}\left(\cdot\right)$) the set containing the columns (resp. the rows) of a matrix.
		The operator $\text{card}\left(\cdot\right)$ returns the number of elements in a set.
		The operator $\text{vec}\left(\cdot\right)$ is the linear transformation which converts a matrix 
		into a column vector, and $\otimes$ denotes the Kronecker product.
		All the sets are denoted with calligraphic notation. Furthermore, we use the
		$\bar{\left(\cdot\right)}$ notation to distinguish the data/digital beamformers from the beam/analog
		ones. The data beamformers are applied over an \emph{effective} channel. The same (bar) notation is used
		to denote the \emph{effective} channel and its second-order statistics.

\section{System Model and Problem Formulation}
	
	For the sake of exposition, this paper ignores inter-cell interference effects and focuses on training 
	and interference within a given cell. Consider a single cell \ac{mMIMO} \ac{BS} equipped with
	$N_{\text{BS}} \gg 1$ antennas which serves (in downlink transmission) $K \ll N_{\text{BS}}$ \acp{UE} 
	with $N_{\text{UE}}$ antennas each. We assume that the \ac{BS} uses linear precoding techniques to 
	process the signals before transmitting to all \acp{UE}. We consider \ac{FDD} operation, i.e. 
	the \ac{DL} and the \ac{UL} channels are not reciprocal.
	
	Before we detail our mathematical model, let us focus on the toy example shown in Fig.~\ref{fig:scen},
	which carries the essence of the intuition behind the proposed trade-off between \emph{i) energy, 
	ii) spatial separability}, and \emph{iii) training overhead}. We explore this trade-off through a
	coordinated beam selection among the \acp{UE}, made \emph{before} the actual training of the beams. 
	
	Consider Fig.~\ref{fig:scen} and the problem of which beams should each \ac{UE} activate and how it 
	affects which beams are lit up at the \ac{BS} and the subsequent training overhead. In a conventional
	strategy, uncoordinated max-\acs{SNR} based beam selection would collect the highest amount of energy 
	but would result in $M_{\text{BS}} = 5$ beams to train at the \ac{BS}. Instead, \ac{UE} $2$ can
	opt for the weaker (non-bold light blue beams) $\mathbf{w}_{2, 1}$ and $\mathbf{w}_{2, 3}$ while \ac{UE} 
	$1$ continues to activate its three beams. Note that this strategy collects less energy, yet it reduces
	the training overhead by $40\%$ as the number of activated beams at the \ac{BS} falls to $M_{\text{BS}} 
	= 3$, since beams $\mathbf{v}_1$, $\mathbf{v}_2$ and $\mathbf{v}_5$ at the \ac{BS} side serve both
	\ac{UE} $1$ and \ac{UE} $2$ and \emph{maintain separability} between them. In the rest of the paper, 
	we are interested in designing a coordinated beam selection algorithm that optimizes this trade-off 
	from a throughput perspective. We introduce now our mathematical model.
					
	\subsection{Channel Estimation with Grid-of-Beams} \label{sec:Channel_Estimation}
		
		We assume that the \ac{GoB} approach is exploited at both the \ac{BS} and \ac{UE} sides. 
		Let us denote the beam codebooks as $\mathcal{B}_{\text{BS}}$ and $\mathcal{B}_{\text{UE}}$, where 
		$\text{card}\left(\mathcal{B}_{\text{BS}}\right) = B_{\text{BS}}$ and 
		$\text{card}\left(\mathcal{B}_{\text{UE}}\right) = B_{\text{UE}}$, used for \ac{GoB} precoding 
		and combining, respectively, as follows:
		\begin{equation}
			\mathcal{B}_{\text{BS}} \triangleq \{ \mathbf{v}_1, \dots, \mathbf{v}_{B_{\text{BS}}} \}, \qquad
			\mathcal{B}_{\text{UE}} \triangleq \{ \mathbf{w}_1, \dots, \mathbf{w}_{B_{\text{UE}}} \},
		\end{equation}
		where $\mathbf{v}_v \in \mathbb{C}^{N_{\text{BS}} \times 1}$, $v \in \{ 1, \dots, B_{\text{BS}} \}$, 
		denotes the $v$-th beamforming vector in $\mathcal{B}_{\text{BS}}$, and 
		$\mathbf{w}_w \in \mathbb{C}^{N_{\text{UE}} \times 1}$, $w \in \{ 1, \dots, B_\text{UE} \}$, 
		denotes the $w$-th beamforming vector in $\mathcal{B}_{\text{UE}}$. 
	 	To lighten the notation, we assume that $\mathcal{B}_{\text{UE}}$ is the same across all the 
	 	UEs\footnote{The algorithms we present in Section~\ref{sec:Dec_Coo_BS_Algos} can be \emph{easily} generalized 
	 	to different codebooks $\mathcal{B}_{\text{UE}}^k ~\forall k \in \{ 1, \dots, K \}$ at the \ac{UE} side.}.
	 	
	 	\begin{figure}[h]
			\centering
			\resizebox{11.57cm}{!}{
			\begin{tikzpicture}
				%% PATHS
				% UE 1 UP
				\draw[orange, densely dashed, line width=1.78] (4,1.74) -- (-.37,3.37);
				\draw[orange, densely dashed, line width=1.78] (-4.8, 1.15) -- (-.37,3.37);
				% UE 1 DOWN First
				\draw[orange, densely dashed] (4,1.74) -- (.16,-3.66);
				\draw[orange, densely dashed] (-4.8, 1.15) -- (.16,-3.66);
				% UE 1 DOWN Second
				\draw[orange, densely dashed, line width=1.78] (4,1.74) -- (-3.5,-1.66);
				\draw[orange, densely dashed, line width=1.78] (-4.8, 1.15) -- (-3.5,-1.66);
				% UE 1 LOS
				\draw[orange, densely dashed, line width=1.78] (4,1.74) -- (-4.8,1.15);
				% UE 2 UP
				\draw[blue, densely dotted] (5.3,-1.78) -- (-.55,3);
				\draw[blue, densely dotted] (-4.8, 1.15) -- (-.55,3);
				% UE 2 DOWN First
				\draw[blue, densely dotted, line width=1.78] (5.3,-1.78) -- (2.96,-3.16);
				\draw[blue, densely dotted, line width=1.78] (-4.8, 1.15) -- (2.96,-3.16);
				% UE 2 DOWN Second
				\draw[blue, densely dotted] (5.3,-1.78) -- (-2.82,-2.36);
				\draw[blue, densely dotted] (-4.8, 1.15) -- (-2.82,-2.36);
				% UE 2 LOS
				\draw[blue, densely dotted, line width=1.78] (5.3,-1.78) -- (-4.8,1.15);
				%% BEAMS
				% BS
				\filldraw[black!71, xshift=1cm, rotate around={26.77:(-4.8-1, 1.15)}] 
				(-4.8, 1.15) ellipse (1cm and .1cm) 
				node[above right, rotate=26.77, xshift=.87cm, yshift=-.077cm]{\small $\mathbf{v}_1$};
				\filldraw[black!71, xshift=1cm, rotate around={4:(-4.8-1, 1.15)}] 
				(-4.8, 1.15) ellipse (1cm and .1cm) 
				node[above right=0cm and 0cm, rotate=4, xshift=.87cm, yshift=-.077cm]{\small $\mathbf{v}_2$};
				\filldraw[black!71, xshift=1cm, rotate around={-15.7:(-4.8-1, 1.15)}] 
				(-4.8, 1.15) ellipse (1cm and .1cm) 
				node[above right=0cm and 0cm, rotate=-15.7, xshift=.87cm, yshift=-.077cm]
				{\small $\mathbf{v}_3$};
				\filldraw[black!71, xshift=1cm, rotate around={-28.7:(-4.8-1, 1.15)}] 
				(-4.8, 1.15) ellipse (1cm and .1cm) 
				node[above right=0cm and 0cm, rotate=-28.7, xshift=.87cm, yshift=-.077cm]
				{\small $\mathbf{v}_4$};
				\filldraw[black!71, xshift=1cm, rotate around={-64.77:(-4.8-1, 1.15)}] 
				(-4.8, 1.15) ellipse (1cm and .1cm) 
				node[above right=0cm and 0cm, rotate=-64.77, xshift=.87cm, yshift=-.077cm]
				{\small $\mathbf{v}_5$};
				% UE 1
				\filldraw[orange, xshift=-1cm, rotate around={-20.27:(4+1, 1.74)}] 
				(4, 1.74) ellipse (1cm and .1cm)
				node[above left, rotate=-20.27, xshift=-.47cm, yshift=-.077cm]{\small $\mathbf{w}_{1, 1}$};
				\filldraw[orange, xshift=-1cm, rotate around={4:(4+1, 1.74)}] (4, 1.74) ellipse (1cm and .1cm)
				node[above left, rotate=4, xshift=-.47cm, yshift=-.077cm]{\small $\mathbf{w}_{1, 2}$};
				\filldraw[orange, xshift=-1cm, rotate around={24.27:(4+1, 1.74)}] 
				(4, 1.74) ellipse (1cm and .1cm)
				node[above left, rotate=24.27, xshift=-.47cm, yshift=-.077cm]{\small $\mathbf{w}_{1, 3}$};
				% UE 2
				\filldraw[blue, xshift=-1cm, opacity=0.27, rotate around={-39.5:(5.3+1, -1.78)}] 
				(5.3, -1.78) ellipse (1cm and .1cm)
				node[opacity=1, above left, rotate=-39.5, xshift=-.47cm, yshift=-.077cm]
				{\small $\mathbf{w}_{2, 1}$};
				\filldraw[blue, xshift=-1cm, rotate around={-16.5:(5.3+1, -1.78)}] 
				(5.3, -1.78) ellipse (1cm and .1cm)
				node[above left, rotate=-16.5, xshift=-.47cm, yshift=-.077cm]{\small $\mathbf{w}_{2, 2}$};
				\filldraw[blue, xshift=-1cm, opacity=0.27, rotate around={4.11:(5.3+1, -1.78)}] 
				(5.3, -1.78) ellipse (1cm and .1cm)
				node[opacity=1, above left, rotate=4.11, xshift=-.47cm, yshift=-.077cm]
				{\small $\mathbf{w}_{2, 3}$};
				\filldraw[blue, xshift=-1cm, rotate around={30.7:(5.3+1, -1.78)}] 
				(5.3, -1.78) ellipse (1cm and .1cm)
				node[above left, rotate =30.7, xshift=-.47cm, yshift=-.077cm]{\small $\mathbf{w}_{2, 4}$};
				%% UE 1
				\filldraw[black] (4,1.37) rectangle +(.3,.3);
				\filldraw[black] (4,1.37) rectangle +(.05,.37);
				\draw (4.21,1.16) node {\small{UE $1$}};
				%% UE 2
				\filldraw[black] (5.3,-2.15) rectangle +(.3,.3);
				\filldraw[black] (5.3,-2.15) rectangle +(.05,.37);
				\draw (5.47,-2.36) node {\small{UE $2$}};
				%% BS
				\draw[very thick] (-5,0.3) -- (-4.8,1.15);
				\draw[very thick] (-4.8,1.15) -- (-4.6,.3);
				\draw[thick] (-4.88,.78) -- (-4.73,.88);
				\draw[thick] (-4.93,.55) -- (-4.7,.7);
				\draw (-4.8, .07) node {BS};
				\filldraw[black] (-4.8,1.15) circle (.1cm); % BS black point on top
				%% CLUSTERS
				% UE 1 UP
				\draw[orange] (-.27,3.1) circle (.53cm) 
				node[black, above right, xshift=.51cm, yshift=.17cm]{\footnotesize{Scatterers}};;
				\filldraw[orange] (-.27, 3.1) circle (.05cm);
				\filldraw[orange] (.03, 3.25) rectangle +(.1, .1);
				\filldraw[orange] (-.42, 3.32) rectangle +(.1, .1);
				% UE 2 UP
				\draw[blue] (-.521,3.2) circle (.53cm);
				\filldraw[blue] (-.521, 3.2) circle (.05cm);
				\filldraw[blue] (-.97, 3.17) rectangle +(.1, .1);
				\filldraw[blue] (-.6, 2.95) rectangle +(.1, .1);
				% UE 1 DOWN
				% First
				\draw[orange] (.28,-3.35) circle (.53cm);
				\filldraw[orange] (.28, -3.35) circle (.05cm);
				\filldraw[orange] (.11, -3.71) rectangle +(.1, .1);
				\filldraw[orange] (.35, -3.11) rectangle +(.1, .1);
				% Second
				\draw[orange] (-3.77,-1.85) circle (.53cm);
				\filldraw[orange] (-3.77, -1.85) circle (.05cm);
				\filldraw[orange] (-3.55, -1.71) rectangle +(.1, .1);
				\filldraw[orange] (-4, -1.58) rectangle +(.1, .1);
				% UE 2 DOWN
				% First
				\draw[blue] (2.64,-3.11) circle (.53cm);
				\filldraw[blue] (2.64, -3.11) circle (.05cm);
				\filldraw[blue] (2.91, -3.21) rectangle +(.1, .1);
				\filldraw[blue] (2.88, -3.44) rectangle +(.1, .1);
				% Second
				\draw[blue] (-3.05,-2.55) circle (.53cm) 
				node[black, above right, xshift=.51cm, yshift=-.57cm]{\footnotesize{Scatterers}};;
				\filldraw[blue] (-3.05, -2.55) circle (.05cm);
				\filldraw[blue] (-3.31, -2.71) rectangle +(.1, .1);
				\filldraw[blue] (-2.87, -2.41) rectangle +(.1, .1);
			\end{tikzpicture}
			}
			\caption{Intuitive toy example with $K = 2$ \acp{UE} highlighting the trade-off between 
			\emph{i) energy} (i.e. activating strong paths), \emph{ii) spatial separability}, and 
			\emph{iii) training overhead} (i.e. lighting up a smaller set of beams at the \ac{BS}). 
			The blue and orange circles represent the multi-path clusters (or scatterers), which might be shared 
			among some \acp{UE}. Stronger paths are marked in bold.}
			\label{fig:scen}
		\end{figure}
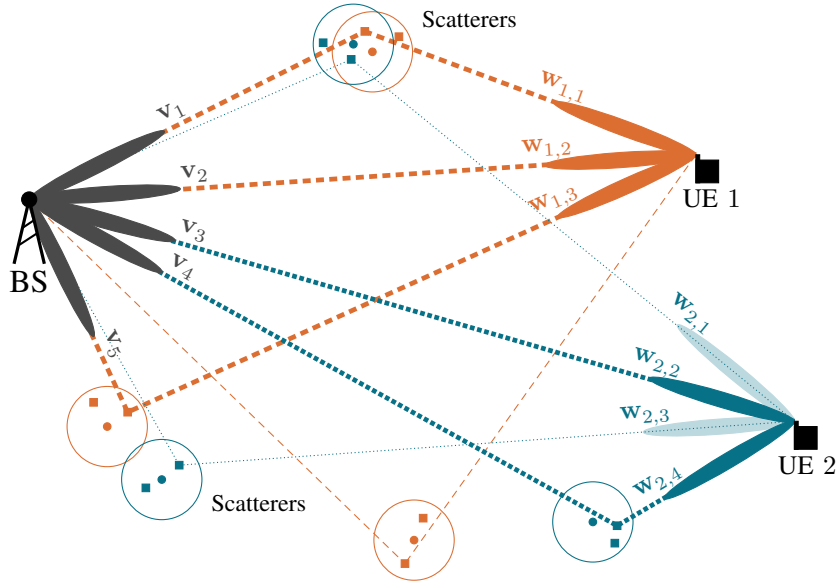
		
		A \ac{NR}-like \acs{OFDM}-based modulation scheme is assumed~\cite{Dahlman2018}. 
		We consider a resource grid consisting of $T$ resource elements.  Among those, $\tau M_{\text{BS}}$
		are allocated to \acp{RS}, and $T - \tau M_{\text{BS}}$ to data, where $M_{\text{BS}}$ denotes
		the number of beams that are trained among the ones in $\mathcal{B}_{\text{BS}}$ and $\tau$ is the
		duration measured in \emph{number of \ac{OFDM} symbols} of their associated \acp{RS} (one \ac{RS} 
		for each beam~\cite{Dahlman2018}, refer to Fig.~\ref{fig:ap_mapping}). 
		The received training signal $\mathbf{Y}_k \in \mathbb{C}^{M_{\text{UE}} \times \tau}$ at the 
		$k$-th UE, where $M_{\text{UE}}$ is the number of activated beams at the \ac{UE} side, can be
		expressed as
		\begin{equation} \label{eq:Y_k}
			\mathbf{Y}_k = 
			\rho \mathbf{W}^{\text{H}}_k \mathbf{H}_k \mathbf{V} \mathbf{S} 
			+ \mathbf{W}^{\text{H}}_k \mathbf{N}_k, 
			\qquad \forall k \in \{1, \dots, K \}
		\end{equation}
		where $\mathbf{S} \in \mathbb{C}^{M_{\text{BS}} \times \tau}$ contains the orthogonal (known)
		\acp{RS}, with $\mathbf{S} \mathbf{S}^{\text{H}} = \mathbf{I}_{M_{\text{BS}}}$, 
		$\mathbf{V} \triangleq \big[ \mathbf{v}_1 \dots \mathbf{v}_{M_{\text{BS}}} \big] \in 
		\mathbb{C}^{N_{\text{BS}} \times M_{\text{BS}}}$ is the normalized training (\ac{GoB}) precoder 
		common to all the \acp{UE}, $\mathbf{H}_k \in \mathbb{C}^{N_{\text{UE}} \times N_{\text{BS}}}$ is 
		the channel between the \ac{BS} and the $k$-th \ac{UE}, with $\text{vec}\left( \mathbf{H}_k \right)
		\sim \mathcal{CN} \left( \mathbf{0}, \mathbf{\Sigma}_k \right)$ and $\mathbf{\Sigma}_k \in
		\mathbb{C}^{N_{\text{BS}} N_{\text{UE}} \times N_{\text{BS}} N_{\text{UE}}}$ the respective channel
		covariance (assumed to be known), and $\mathbf{W}_k \triangleq 
		\big[ \mathbf{w}_{k, 1} \dots \mathbf{w}_{k, M_{\text{UE}}} \big] 
		\in \mathbb{C}^{N_\text{UE} \times M_\text{UE}}$ is the training combiner at the $k$-th \ac{UE}. 
		Note that both $\mathbf{V}$ and $\mathbf{W}_k ~\forall k$ contain beamformers belonging to the
		predefined \ac{GoB} codebooks $\mathcal{B}_{\text{BS}}$ and $\mathcal{B}_{\text{UE}}$. 
		The matrix $\mathbf{N}_k \in \mathbb{C}^{N_{\text{UE}} \times \tau}$, whose elements are i.i.d. 
		$\mathcal{CN}\left(0, \sigma_n^2\right)$, denotes the receiver noise at the $k$-th \ac{UE}, while 
		$\rho \triangleq \sqrt{\frac{P}{T}}$, where $P$ is the total transmit power
		available at the \ac{BS} in the considered coherent (over both time and sub-carriers) frame.
		
		Following the training stage, the \acp{UE} are able to estimate their instantaneous \ac{GoB} effective channels, defined as
		\begin{equation}
			\bar{\mathbf{H}}_k \triangleq \mathbf{W}^{\text{H}}_k \mathbf{H}_k \mathbf{V} \in \mathbb{C}^{M_{\text{UE}} 
			\times M_{\text{BS}}}, \qquad \forall k \in \{1, \dots, K \}
		\end{equation}
		and whose covariance is denoted with 
		$\bar{\mathbf{\Sigma}}_k \in \mathbb{C}^{M_{\text{BS}} M_{\text{UE}} \times M_{\text{BS}} M_{\text{UE}}}, 
		~\forall k \in \{1, \dots, K \}$.
		\begin{remark}
			With respect to the channel estimation, we assume that each \ac{UE} has (at least) $M_{\text{UE}}$ independent
			\acs{RF} chains available. In particular, such assumption implies that each \ac{UE} can process the incoming
			training signal at the receive beams $\big[ \mathbf{w}_{k, 1} \dots \mathbf{w}_{k, M_{\text{UE}}} \big] $ in parallel. \qed
		\end{remark}
		
		We introduce now the block diagonal matrix $\mathbf{W} \in \mathbb{C}^{K N_{\text{UE}} \times K M_{\text{UE}}}$ containing 
		all the \ac{GoB} combiners $\mathbf{W}_k ~\forall k \in \{ 1, \dots, K \}$, as follows:
		\begin{equation} \label{eq:overall_GoB_combiner}
			\mathbf{W} \triangleq \begin{bmatrix}
				\mathbf{W}_1 & & & \mathbf{0} \\
				& & \ddots & \\
				\mathbf{0} & & & \mathbf{W}_K
			\end{bmatrix}.
		\end{equation}
		The entire multi-user effective channel matrix $\bar{\mathbf{H}} \in \mathbb{C}^{K M_{\text{UE}} \times M_{\text{BS}}}$ can 
		then be expressed as
		\begin{equation}
			\bar{\mathbf{H}} \triangleq \mathbf{W}^{\text{H}} \mathbf{H} \mathbf{V},
		\end{equation}
		where $\mathbf{H} \triangleq \big[ \mathbf{H}_1^{\text{T}} \dots \mathbf{H}_K^{\text{T}} \big]^{\text{T}} 
		\in \mathbb{C}^{KN_{\text{UE}} \times N_{\text{BS}}}$ is the overall multi-user channel.
		
		To close the \ac{CSI} acquisition loop, each \ac{UE} feeds back its estimated effective channel to the \ac{BS}. 
		As a consequence, the \ac{BS} obtains an estimate $\hat{\bar{\mathbf{H}}} \in \mathbb{C}^{K M_{\textnormal{UE}} 
		\times M_{\textnormal{BS}}}$ of the multi-user effective channel $\bar{\mathbf{H}}$ which can be used to design 
		the \ac{mMIMO} data precoder. In this work, we assume that the UEs use the popular \ac{LMMSE} estimator, 
		for which the effective channel estimate $\hat{\bar{\mathbf{H}}}_k \in \mathbb{C}^{M_{\textnormal{UE}} \times M_{\textnormal{BS}}}$
		at the $k$-th \ac{UE} reads as follows~\cite{Kay1993}:
		\begin{equation} \label{eq:mmse_estimate}
			\textnormal{vec}\big(\hat{\bar{\mathbf{H}}}_k\big) = \rho \bar{\mathbf{\Sigma}}_k \mathbf{A}^{\textnormal{H}} 
			\left( \rho^2 \mathbf{A}\bar{\mathbf{\Sigma}}_k\mathbf{A}^{\textnormal{H}} + 
			\sigma_n^2\mathbf{\Gamma}\mathbf{\Gamma}^{\textnormal{H}} \right)^{-1} 
			\textnormal{vec}\left( \mathbf{Y}_k \right),
		\end{equation}
		where $\mathbf{A} \triangleq \left(\mathbf{S}^{\textnormal{T}} \otimes \mathbf{I}_{M_{\textnormal{UE}}} \right)
		\in \mathbb{C}^{\tau M_{\textnormal{UE}} \times M_{\textnormal{BS}} M_{\textnormal{UE}}}$ and
		$\mathbf{\Gamma} \triangleq \left( \mathbf{I}_{\tau} \otimes \mathbf{W}^{\textnormal{H}}_k \right) \in
		\mathbb{C}^{\tau M_{\textnormal{UE}} \times \tau N_{\textnormal{UE}}}$.

		The related channel estimation error vector $\mathbf{e}_k = 
		\textnormal{vec}\big(\bar{\mathbf{H}}_k\big) - \textnormal{vec}\big(\hat{\bar{\mathbf{H}}}_k\big)$ at the $k$-th \ac{UE} 
		has zero mean elements~\cite{Kay1993} and associated covariance matrix as given in the next lemma.
		\vspace{-0.21cm}
		\begin{lemma}
			The covariance $\mathbf{\Sigma}_{\mathbf{e}_k} \in \mathbb{C}^{M_{\textnormal{BS}} M_{\textnormal{UE}}\times 
			M_{\textnormal{BS}} M_{\textnormal{UE}}}$ of the \ac{LMMSE} channel estimation error at the $k$-th \ac{UE} can be expressed as follows:
			\begin{equation}  \label{eq:lem_lmmse}
				\mathbf{\Sigma}_{\mathbf{e}_k} = \left( 
				\bar{\mathbf{\Sigma}}_k^{-1} + \kappa
				\mathbf{A}^{\textnormal{H}} \left( \mathbf{\Gamma} \mathbf{\Gamma}^{\textnormal{H}} \right)^{-1} \mathbf{A}
				\right)^{-1},
			\end{equation}
			having defined the scalar $\kappa \triangleq \rho^2 / \sigma_n^2$.
		\end{lemma}
		\begin{proof}
			From the definition of the error covariance 
			$\mathbf{\Sigma}_{\mathbf{e}_k} \triangleq \mathbb{E}_{\mathbf{H}_k} \left[ \mathbf{e}_k \mathbf{e}_k^{\textnormal{H}}\right]$, we have:
			\begin{align}
				\mathbf{\Sigma}_{\mathbf{e}_k} &=
				\mathbf{\bar{\Sigma}}_k - \rho \bar{\mathbf{\Sigma}}_k \mathbf{A}^{\textnormal{H}} 
				\left( \rho^2 \mathbf{A}\bar{\mathbf{\Sigma}}_k\mathbf{A}^{\textnormal{H}} + 
				\sigma_n^2\mathbf{\Gamma}\mathbf{\Gamma}^{\textnormal{H}} \right)^{-1}
				\rho \mathbf{A} \mathbf{\bar{\Sigma}}_k \nonumber \\ 
				&\stackrel{(a)}{=} \mathbf{\bar{\Sigma}}_k - \left( 
				\bar{\mathbf{\Sigma}}_k^{-1} + 
				\kappa \mathbf{A}^{\textnormal{H}} \left( \mathbf{\Gamma} \mathbf{\Gamma}^{\textnormal{H}} \right)^{-1} \mathbf{A} 
				\right)^{-1} \kappa \mathbf{A}^{\textnormal{H}} \left( \mathbf{\Gamma} \mathbf{\Gamma}^{\textnormal{H}} \right)^{-1} 
				\mathbf{A} \bar{\mathbf{\Sigma}}_k,
			\end{align}
			where $(a)$ is due to the \emph{Woodbury identity}.
			We can then rewrite the error covariance as
			\begin{align}
				\mathbf{\Sigma}_{\mathbf{e}_k} &= 
				\mathbf{\bar{\Sigma}}_k - \mathbf{D}^{-1} 
				\kappa \mathbf{A}^{\textnormal{H}} \left( \mathbf{\Gamma} \mathbf{\Gamma}^{\textnormal{H}} \right)^{-1} 
				\mathbf{A} \bar{\mathbf{\Sigma}}_k \nonumber \\
				&= \mathbf{D}^{-1} \left( \mathbf{D} \mathbf{\bar{\Sigma}}_k - \kappa \mathbf{A}^{\textnormal{H}} 
				\left( \mathbf{\Gamma} \mathbf{\Gamma}^{\textnormal{H}} \right)^{-1} \mathbf{A} \bar{\mathbf{\Sigma}}_k \right) \nonumber \\
				&= \mathbf{D}^{-1},
			\end{align}
			where $\mathbf{D} = \left( \bar{\mathbf{\Sigma}}_k^{-1} + \kappa \mathbf{A}^{\textnormal{H}} 
			\left( \mathbf{\Gamma} \mathbf{\Gamma}^{\textnormal{H}} \right)^{-1} \mathbf{A} \right)$, as given in \eqref{eq:lem_lmmse}.
		\end{proof}

	\subsection{Data Signal Model} \label{sec:data_signal_model}
		
		The data transmission phase (over the effective channels) follows the training and \ac{UE} feedback stages. 
		Let us consider a single resource element, and denote with 
		$\mathbf{x}_k \triangleq \big[ x_{1, 1} \dots  x_{1, L_k} \big] \in \mathbb{C}^{L_k \times 1}$ the data vector transmitted to the $k$-th \ac{UE}.
		Thus, $\mathbf{x} \triangleq \big[ \mathbf{x}_1 \dots  \mathbf{x}_K \big] \in \mathbb{C}^{L \times 1}$ is the overall data vector, where 
		$L \triangleq \sum_k L_k$ is the total number of transmitted data symbols and $\mathbb{E}[\mathbf{x} \mathbf{x}^{\text{H}}] = \mathbf{I}_L$.
		The received data signal $\hat{\mathbf{x}}_k$ at the $k$-th UE can be expressed as
		\begin{align} \label{eq:x_hat_k}
			\hat{\mathbf{x}}_k &= \rho \bar{\mathbf{W}}^{\text{H}}_k \bar{\mathbf{H}}_k \bar{\mathbf{V}} \mathbf{x} + 
			\bar{\mathbf{W}}^{\text{H}}_k \bar{\mathbf{n}}_k, \qquad \forall k \in \{1, \dots, K \} \nonumber \\
			&= \rho \bar{\mathbf{W}}^{\text{H}}_k \bar{\mathbf{H}}_k \bar{\mathbf{V}}_k \mathbf{x}_k +
			\sum_{j \neq k} \rho \bar{\mathbf{W}}^{\text{H}}_k \bar{\mathbf{H}}_k \bar{\mathbf{V}}_j \mathbf{x}_j +
			\bar{\mathbf{W}}^{\text{H}}_k \bar{\mathbf{n}}_k,
		\end{align}
		where $\bar{\mathbf{V}} \triangleq \big[ \bar{\mathbf{V}}_1 \dots \bar{\mathbf{V}}_K \big] \in \mathbb{C}^{M_{\text{BS}} \times L}$ is the
		normalized \ac{mMIMO} (digital) data precoder, with $\bar{\mathbf{V}}_k \triangleq \big[ \bar{\mathbf{v}}_{k, 1} \dots \bar{\mathbf{v}}_{k, L_k}
		\big]$, $\bar{\mathbf{H}}_k$ is the effective channel between the \ac{BS} and the $k$-th \ac{UE} after \ac{GoB} precoding and combining, 
		$\bar{\mathbf{W}}_k \in \mathbb{C}^{M_{\text{UE}} \times L_k}$ is the \ac{mMIMO} (digital) data combiner at the $k$-th \ac{UE},
		and $\bar{\mathbf{n}}_k \triangleq \mathbf{W}^{\text{H}}_k \mathbf{n}_k \in \mathbb{C}^{M_{\text{UE}} \times 1}$ denotes the filtered
		receiver noise at the $k$-th \ac{UE}.
		
		The instantaneous \ac{SE} $\mathcal{R}_k(\mathbf{V}, \bar{\mathbf{V}}, \mathbf{W}, \bar{\mathbf{W}})$ relative to the $k$-th \ac{UE} can then be expressed as follows:
		\begin{equation} \label{eq:SINR_gen}
			\mathcal{R}_k\left(\mathbf{V}, \bar{\mathbf{V}}, \mathbf{W}, \bar{\mathbf{W}}\right) \triangleq \log_2 \det \left(
			 \mathbf{I}_{L_k} + \rho^2 \bar{\mathbf{K}}_k^{-1} \bar{\mathbf{W}}_k^{\text{H}} \bar{\mathbf{H}}_k \bar{\mathbf{V}}_k
			 \bar{\mathbf{V}}_k^{\text{H}} \bar{\mathbf{H}}_k^{\text{H}} \bar{\mathbf{W}}_k \right),
		\end{equation}
		where	$\bar{\mathbf{K}}_k \triangleq \rho^2 \sum_{j \neq k} \bar{\mathbf{W}}_k^{\text{H}} \bar{\mathbf{H}}_k \bar{\mathbf{V}}_j 
		\bar{\mathbf{V}}_j^{\text{H}} \bar{\mathbf{H}}_k^{\text{H}} \bar{\mathbf{W}}_k + 
		\sigma_n^2  \bar{\mathbf{W}}_k^{\text{H}} \mathbf{W}_k^{\text{H}} \mathbf{W}_k \bar{\mathbf{W}}_k$ is the interference plus noise
		covariance relative to the $k$-th \ac{UE}, and where we recall that the dependence on $\mathbf{V}$ and $\mathbf{W}$ is because
		$\bar{\mathbf{H}} \triangleq \mathbf{W}^{\text{H}} \mathbf{H} \mathbf{V}$.
		
	\subsection{Optimal Precoders and Combiners} \label{sec:Coo_Beam_Rep}
		
		In order to design a processing scheme which achieves the optimal effective network throughput, the mutual optimization of the (constrained)
		\ac{GoB} and (unconstrained) \ac{mMIMO} data beamformers should be considered. Let us first define the overall training overhead as follows.
		\begin{definition} \label{def:train_ov_v}
			Let $\mathbf{V} \in \mathbb{C}^{N_{\textnormal{BS}} \times M_{\textnormal{BS}}}$ be the \ac{GoB} precoder at the \ac{BS}.
			The training overhead $\omega\left(\mathbf{V}\right) \in [0, 1]$ in terms of pilot resource elements is defined as follows:
			\begin{equation} \label{eq:training_overhead_v}
				\omega\left(\mathbf{V}\right) \triangleq \frac{\tau}{T} \textnormal{card}\left(\textnormal{col}\left(\mathbf{V}\right)\right).
			\end{equation}
			% where we recall that $\tau$ is the duration in \ac{OFDM} symbols of the beamformed \acp{RS}.
		\end{definition}
		Note that the training overhead depends on how the \ac{GoB} precoder $\mathbf{V}$ is designed. 
		$\text{Col}\left(\mathbf{V}\right)$ consist indeed of the beams to train in the channel estimation phase 
		(refer to Eq. \eqref{eq:Y_k} and Fig.~\ref{fig:ap_mapping}).
		
		Therefore, the achievable effective network throughput $\mathcal{R}$ can be expressed as
		\begin{align} \label{eq:ent}
			\mathcal{R}\left(\mathbf{V}, \bar{\mathbf{V}}, \mathbf{W}, \bar{\mathbf{W}}\right) \triangleq \left(1-\omega\left(\mathbf{V}\right)\right)
			\sum_{k=1}^K \mathcal{R}_k\left(\mathbf{V}, \bar{\mathbf{V}}, \mathbf{W}, \bar{\mathbf{W}}\right).
		\end{align}
		
		The optimal beamformers $\left(\mathbf{V}^*, \bar{\mathbf{V}}^*, \mathbf{W}^*, \bar{\mathbf{W}}^*\right)$ are then found as follows:
		\begin{align} \label{eq:gen_opt} \tag{P$\star$}
			\left(\mathbf{V}^*, \bar{\mathbf{V}}^*, \mathbf{W}^*, \bar{\mathbf{W}}^*\right) = 
			&\argmax_{\mathbf{V}, \bar{\mathbf{V}}, \mathbf{W}, \bar{\mathbf{W}}}
			~\mathbb{E}_{\mathbf{H}} \Big[ \mathcal{R}\left(\mathbf{V}, \bar{\mathbf{V}}, \mathbf{W}, \bar{\mathbf{W}}\right) \Big], \\
			& \text{subject to} ~\text{col}\left(\mathbf{V}\right) \in \mathcal{B}_{\text{BS}} \nonumber \\
			& \phantom{\text{subject to}} ~\text{col}\left(\mathbf{W}_k\right) \in \mathcal{B}_{\text{UE}}, ~\forall k = \{ 1, \dots, K \}. \nonumber
		\end{align}		
		Finding the global optimum for the optimization problem \eqref{eq:gen_opt} is not trivial and often found to be intractable, even without
		considering the pre-log factor relative to the training overhead~\cite{Ayach2014, Alkhateeb2015}. A common and viable approach consists in
		decoupling the design, as the \ac{GoB} beamformers can be optimized through long-term statistical information, whereas the \ac{mMIMO} data
		beamformers can depend on the instantaneous \ac{CSI}~\cite{Alkhateeb2015}. The same approach is followed in this work.
		In particular, we consider two different timescales:
		\begin{itemize}
			\item{\textbf{Small timescale} \emph{(channel coherence time)}: within which the instantaneous channel realization $\mathbf{H}_k, ~\forall k$
			is assumed to be constant and a single training phase is carried out;}
			\item{\textbf{Large timescale} \emph{(beam coherence time)}: within which the covariance matrices $\mathbf{\Sigma}_k, ~\forall k$ 
			are assumed to be constant and the \ac{GoB} beamformers are designed \emph{(beam selection)}.}
		\end{itemize}
		In the following section, we will focus on the design of the \ac{mMIMO} data precoder and combiners with given multi-user effective channel 
		$\bar{\mathbf{H}}$. Later, the design of the long-term \ac{GoB} beamformers will be considered assuming fixed \ac{mMIMO} data beamformers.
		
		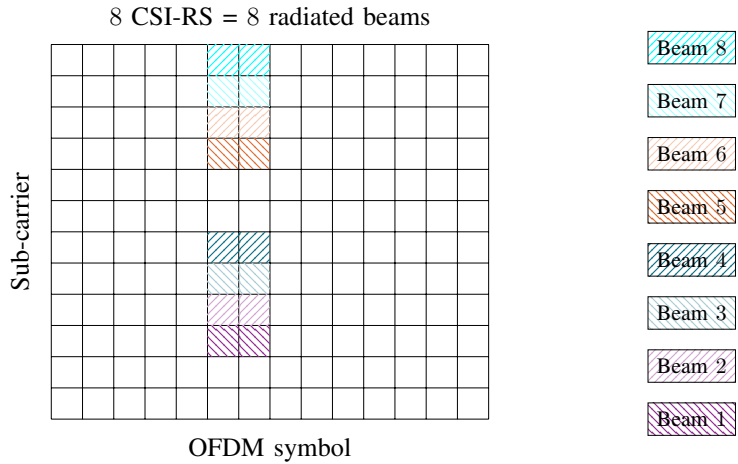
\begin{figure}[h]
			\centering
			 \resizebox{10cm}{!}{
			 \begin{tikzpicture}[every node/.style={minimum size=.5cm-\pgflinewidth, outer sep=0pt}]
				% REs
    				\draw[step=0.5cm,color=black] (-10,-3) grid (-3,3);
    				\node[pattern=north west lines, pattern color=violet!85] at (-7.25,-1.75) {};
    				\node[pattern=north west lines, pattern color=violet!85] at (-6.75,-1.75) {};
    				\node[pattern=north east lines, pattern color=violet!35] at (-7.25,-1.25) {};
    				\node[pattern=north east lines, pattern color=violet!35] at (-6.75,-1.25) {};
    				\node[pattern=north west lines, pattern color=blue!35] at (-7.25,-0.75) {};
    				\node[pattern=north west lines, pattern color=blue!35] at (-6.75,-0.75) {};
    				\node[pattern=north east lines, pattern color=blue] at (-7.25,-0.25) {};
    				\node[pattern=north east lines, pattern color=blue] at (-6.75,-0.25) {};
    				\node[pattern=north west lines, pattern color=orange] at (-7.25,1.25) {};
    				\node[pattern=north west lines, pattern color=orange] at (-6.75,1.25) {};
    				\node[pattern=north east lines, pattern color=orange!35] at (-7.25,1.75) {};
    				\node[pattern=north east lines, pattern color=orange!35] at (-6.75,1.75) {};
    				\node[pattern=north west lines, pattern color=cyan!35] at (-7.25,2.25) {};
    				\node[pattern=north west lines, pattern color=cyan!35] at (-6.75,2.25) {};
    				\node[pattern=north east lines, pattern color=cyan] at (-7.25,2.75) {};
    				\node[pattern=north east lines, pattern color=cyan] at (-6.75,2.75) {};
    				% Legend
    				\node[pattern=north west lines, pattern color=violet!85, draw=black] at (.25,-3) {\small Beam $1$};
    				\node[pattern=north east lines, pattern color=violet!35, draw=black] at (.25,-2.15) {\small Beam $2$};
    				\node[pattern=north west lines, pattern color=blue!35, draw=black] at (.25,-1.3) {\small Beam $3$};
    				\node[pattern=north east lines, pattern color=blue, draw=black] at (.25,-0.45) {\small Beam $4$};
    				\node[pattern=north west lines, pattern color=orange, draw=black] at (.25,0.4) {\small Beam $5$};
    				\node[pattern=north east lines, pattern color=orange!35, draw=black] at (.25,1.25) {\small Beam $6$};
    				\node[pattern=north west lines, pattern color=cyan!35, draw=black] at (.25,2.1) {\small Beam $7$};
    				\node[pattern=north east lines, pattern color=cyan, draw=black] at (.25,2.95) {\small Beam $8$};
    				% Axis
    				\node at (-6.5, -3.5){OFDM symbol};
    				\node[rotate=90] at (-10.5, 0){Sub-carrier};
    				\node at (-6.5, 3.45){$8$ \ac{CSI}-{RS} = $8$ radiated beams};
			 \end{tikzpicture}
			 }
			\caption{CSI-RS locations in a \ac{DL} \ac{NR} resource block for $M_{\text{BS}} = 8$. 
			When \ac{GoB} precoding is used, the radiated beams are mapped to one precoded \ac{CSI}-\ac{RS} each, sent over 
			$\tau M_{\text{BS}}$ non-overlapping resource elements on \emph{distinct} antenna ports~\cite{3GPP2018c}.
			Therefore, less resource elements are available for transmitting data to the \acp{UE}, leading to throughput degradation.}
			\label{fig:ap_mapping}
		\end{figure}

\section{Data Beamformers Design} \label{sec:Data_BF_Design}
		
		Since we consider multi-beam processing at the \ac{UE} side, i.e. $M_{\text{UE}} > 1$, 
		the complete diagonalization of the effective channel $\bar{\mathbf{H}}$ at the \ac{BS} side 
		is suboptimal~\cite{Spencer2004}. The \ac{BD} approach is a popular method to design near-optimal
		beamformers that eliminate the multi-user interference in such scenarios. In particular, the 
		\ac{mMIMO} data precoder $\bar{\mathbf{V}}$ at the \ac{BS} side aims to produce a block-diagonal 
		$\bar{\mathbf{H}} \bar{\mathbf{V}}$ where no multi-user interference is experienced. 
		The eventual remaining inter-stream interference can then be suppressed at the \ac{UE} side through 
		a proper combining operation. 
		In this section, we review the complete procedure to perform the \ac{BD}~\cite{Spencer2004}, 
		which will allow for a simplified \ac{SE} expression depending on the long-term \ac{GoB} beamformers
		\emph{only}.
		
		To ensure a block-diagonal $\bar{\mathbf{H}} \bar{\mathbf{V}}$, the precoding matrix $\bar{\mathbf{V}}_k$ has to be designed such that
		\begin{equation} \label{eq:BD_constraint}
			\bar{\mathbf{H}}_j \bar{\mathbf{V}}_k = \mathbf{0}, ~\forall j \neq k.
		\end{equation}
		Introducing the matrix $\bar{\mathbf{H}}_{/k} \in \mathbb{C}^{(K-1)M_{\text{UE}} \times M_{\text{BS}}}$ as
		\begin{equation}
			\bar{\mathbf{H}}_{/k} \triangleq \big[ 
			\bar{\mathbf{H}}_1^{\text{T}} \dots \bar{\mathbf{H}}_{k-1}^{\text{T}} \bar{\mathbf{H}}_{k+1}^{\text{T}} \dots \bar{\mathbf{H}}_K^{\text{T}} 
			\big]^{\text{T}},
		\end{equation}
		the condition in \eqref{eq:BD_constraint} is enforced by letting $\bar{\mathbf{V}}_k$ lie in $\text{null}\left(\bar{\mathbf{H}}_{/k}\right)$.
		Whenever $\text{card}\left(\text{null}\left(\bar{\mathbf{H}}_{/k}\right)\right) \neq 0$, which holds when 
		$\text{rank}\left( \bar{\mathbf{H}}_{/k} \right) < M_{\text{BS}}$, the \ac{BS} can send (multi-user) interference-free data to the $k$-th \ac{UE}.
		
		As a first step, the \ac{SVD} is performed on $\bar{\mathbf{H}}_{/k}$. We can write
		\begin{equation} \label{eq:svd_h_no_k}
			\bar{\mathbf{H}}_{/k} = 
			\bar{\mathbf{U}}_{/k} \bar{\mathbf{S}}_{/k} \left[ \bar{\mathbf{M}}_{/k}^{(1)} \bar{\mathbf{M}}_{/k}^{(0)} \right]^{\text{H}},
		\end{equation}
		where $\bar{\mathbf{M}}_{/k}^{(1)}$ contains the first $\bar{M}_{/k} \triangleq \text{rank}\left(\bar{\mathbf{H}}_{/k}\right)$ right 
		singular vectors of $\bar{\mathbf{H}}_{/k}$, while $\bar{\mathbf{M}}_{/k}^{(0)}$ contains the last $(M_{\text{BS}} - \bar{M}_{/k})$ ones.
		Thus, we know that
		\begin{equation}
			\bar{\mathbf{H}}_j \bar{\mathbf{M}}_{/k}^{(0)} = \mathbf{0}, ~\forall j \neq k.
		\end{equation}
		The \ac{BD} of the overall multi-user effective channel $\bar{\mathbf{H}}$ can then be expressed as
		\begin{equation}
			\bar{\mathbf{H}}_{\text{BD}} = \begin{bmatrix}
				\bar{\mathbf{H}}_1 \bar{\mathbf{M}}_{/1}^{(0)} & & & \mathbf{0} \\
				& & \ddots & \\
				\mathbf{0} & & & \bar{\mathbf{H}}_K \bar{\mathbf{M}}_{/K}^{(0)}
			\end{bmatrix}.
		\end{equation}
		To achieve optimal \ac{SE}, further \ac{SVD}-based processing is carried out~\cite{Spencer2004}. 
		Since $\bar{\mathbf{H}}_{\text{BD}}$ is block diagonal, we can perform an individual \ac{SVD} for 
		each \ac{UE} rather than decomposing the overall large matrix $\bar{\mathbf{H}}_{\text{BD}}$. 
		In particular, we can write
		\begin{equation} \label{eq:BD_dep_on_VW}
			\bar{\mathbf{H}}_k \bar{\mathbf{M}}_{/k}^{(0)} =
			\left[ \bar{\mathbf{U}}_{k}^{(1)} \bar{\mathbf{U}}_{k}^{(0)} \right] 
			\begin{bmatrix} \bar{\mathbf{S}}_{k} & \mathbf{0} \\ \mathbf{0} & \mathbf{0} \end{bmatrix} 
			\left[ \bar{\mathbf{M}}_{k}^{(1)} \bar{\mathbf{M}}_{k}^{(0)} \right]^{\text{H}}.
		\end{equation}
		The product $\bar{\mathbf{M}}_{/k}^{(0)} \bar{\mathbf{M}}_{k}^{(1)}$ produces an orthogonal basis 
		with dimension $L_k \triangleq \text{rank}\left(\bar{\mathbf{H}}_k \bar{\mathbf{M}}_{/k}^{(0)}\right)$
		and can be used as the multi-user interference-nulling data precoder for the $k$-th \ac{UE}, i.e. 
		$\bar{\mathbf{V}}_k = \bar{\mathbf{M}}_{/k}^{(0)} \bar{\mathbf{M}}_{k}^{(1)}$. 
		In order to send interference-free data to the $k$-th \ac{UE}, 
		$\text{rank}\left(\bar{\mathbf{H}}_k \bar{\mathbf{M}}_{/k}^{(0)}\right) \ge 1$ is needed. 
		The receive data combiner $\bar{\mathbf{W}}_k$ relative to the $k$-th \ac{UE} is then designed as 
		$\bar{\mathbf{W}}_k = \bar{\mathbf{U}}_{k}^{(1)}$.
		\begin{lemma} \label{lemma:interference_cancel}
			The condition $\text{rank}\left(\bar{\mathbf{H}}_k \bar{\mathbf{M}}_{/k}^{(0)}\right) \ge 1$ is respected when there exists at least one 
			vector in $\text{row}\left(\bar{\mathbf{H}}_k\right)$ that is \ac{LI} of $\text{row}\left(\bar{\mathbf{H}}_{/k}\right)$. 
		\end{lemma}
		\begin{proof}
			Let us assume that $\exists~\mathbf{k} \in \text{row}\left(\bar{\mathbf{H}}_k \right)$ that is \ac{LI} of 
			$\text{row}\left(\bar{\mathbf{H}}_{/k}\right)$. Then, since $\bar{\mathbf{M}}_{/k}^{(0)}$ is a basis for 
			$\text{null}\left(\bar{\mathbf{H}}_{/k}\right)$, we have $\mathbf{k} \bar{\mathbf{M}}_{/k}^{(0)} \neq \mathbf{0}$. 
			Therefore, $\text{rank}\left(\bar{\mathbf{H}}_k \bar{\mathbf{M}}_{/k}^{(0)}\right) \ge 1$.
		\end{proof}
		Note that inverting the entire $\bar{\mathbf{H}}$ at the \ac{BS} side through e.g. \ac{ZF} precoding requires that each vector in 
		$\text{row}\left(\bar{\mathbf{H}}_k\right)$ is \ac{LI} of $\text{row}\left(\bar{\mathbf{H}}_{/k}\right)$. 
		The \ac{BD} approach offers thus more freedom for designing the \ac{GoB} beamformers $\mathbf{V}$ and $\mathbf{W}$, although a 
		higher value for $\text{rank}\left(\bar{\mathbf{H}}_k \bar{\mathbf{M}}_{/k}^{(0)}\right)$ would still be beneficial for increasing the \ac{SE} 
		(more available streams for the $k$-th \ac{UE}).
		\begin{proposition}[\cite{Spencer2004}]
			When all the interference cancellation conditions are met, the instantaneous \ac{SE} after \ac{BD} precoding
			$\mathcal{R}^{\textnormal{BD}}_k(\mathbf{V}, \mathbf{W})$ relative to the $k$-th \ac{UE} can be written as follows:
			\begin{equation} \label{eq:SINR_gen_BD}
				\mathcal{R}_k^{\textnormal{BD}} \left( \mathbf{V}, \mathbf{W} \right) \triangleq \log_2 \det \left(
				 \mathbf{I}_{L_k} + \kappa \bar{\mathbf{S}}_k^{\textnormal{H}} \bar{\mathbf{S}}_k \right),
			\end{equation}
			where the dependence on $\mathbf{V}$ and $\mathbf{W}$ is hidden in the linear transformation \eqref{eq:BD_dep_on_VW}.
		\end{proposition}
		Therefore, fixing \ac{BD} as the \ac{mMIMO} data precoder allows to reformulate \eqref{eq:gen_opt} as a long-term joint transmit-receive 
		beam selection problem, where the optimum \ac{GoB} beamformers $\left(\mathbf{V}^{\text{(P0)}}, \mathbf{W}^{\text{(P0)}}\right)$ are found 
		as follows:
		\begin{align} \label{eq:gen_opt_after_BD} \tag{P0}
			\left(\mathbf{V}^{\text{(P0)}}, \mathbf{W}^{\text{(P0)}}\right) = 
			&\argmax_{\mathbf{V}, \mathbf{W}}
			~\mathbb{E}_{\mathbf{H}} \Big[ 
			\left(1-\omega\left(\mathbf{V}\right)\right) \sum_{k=1}^K \mathcal{R}^{\text{BD}}_k \left(\mathbf{V}, \mathbf{W} \right) \Big], \\
			& \text{subject to} ~\text{col}\left(\mathbf{V}\right) \in \mathcal{B}_{\text{BS}} \nonumber \\
			& \phantom{\text{subject to}} ~\text{col}\left(\mathbf{W}_k\right) \in \mathcal{B}_{\text{UE}}, ~\forall k = \{ 1, \dots, K \}. \nonumber
		\end{align}		
		The problem \eqref{eq:gen_opt_after_BD} is a discrete optimization problem with a non-convex objective function.  
		The solution for this class of problems is often hard to find and requires combinatorial search, alternating minimization algorithms or 
		relaxation techniques, which are however demanding to put into practice. In this work, we aim instead to design heuristic beam selection 
		algorithms. In the next section, we will thus deal with the design of the long-term \ac{GoB} beamformers $\mathbf{V}$ and $\mathbf{W}$.
				
\section{Grid-of-Beams Beamformers Design} \label{sec:Coo_GoB_Design}
	
	In general, it can be seen through inspecting the objective function in \eqref{eq:gen_opt_after_BD} that designing proper \ac{GoB} beamformers 
	$\mathbf{V}$ and $\mathbf{W}$ implies \emph{i)} harvesting large effective channel gain, \emph{ii)} avoiding catastrophic multi-user interference,
	and \emph{iii)} minimizing the training overhead. In this section, we investigate such conditions in detail so as to set the requirements for 
	an effective \ac{GoB} beamformers design. In particular, for each condition, we introduce a related beam selection optimization problem which
	approximates \eqref{eq:gen_opt_after_BD} and whose practical implementation will be discussed in Section~\ref{sec:Dec_Coo_BS_Algos}.
	To this end, we define the notion of relevant channel components and take a closer look at the beam reporting procedure defined in the current
	$5$G \ac{NR} specifications. Furthermore, we highlight the role of coordinating \acp{UE} in reducing the multi-user interference and the training
	overhead in the considered \ac{FDD} \ac{mMIMO} scenario.

	\subsection{Harvesting Large Effective Channel Gain} \label{sec:RelBeams}
		
		In the classical \ac{GoB} implementation all the beams in the grid are trained regardless of their actual relevance, i.e. 
		$M_{\text{BS}} = B_\text{BS}$. As pointed out in Section~\ref{sec:Intro}, such an operating mode is feasible for small \acp{GoB} only (refer to
		Fig.~\ref{fig:ap_mapping}), although employing a small \ac{GoB}, in turn, leads to a high performance loss~\cite{Zirwas2016}. 
		In order to avoid exchanging performance for overhead, the intuition is to use a large \ac{GoB} and leverage the knowledge of the long-term 
		statistical information to train a few \emph{(accurately)} selected beams to train, so as to keep $\omega = \left(\tau/T\right) M_{\text{BS}}$ 
		small. In particular, in order to gather as much beamforming gain as possible, the idea is to capitalize on the so-called 
		\emph{relevant channel components}, whose number depends on the propagation environment.
		
		\begin{remark}
			This intuition has been exploited, to a large extent, to optimize single-user mmWave communications. Owing to the sparse mmWave radio
			environment, few beams are enough to obtain an accurate and profitable low-dimensional representation of the actual channel~\cite{Heath2016}.
			\qed
		\end{remark}
		
		\begin{definition}
			We define the set $\mathcal{M}_k$ containing the relevant channel components (or relevant beam pairs) of the $k$-th \ac{UE} as follows:
			\begin{equation} \label{eq:rel_channel_comp}
				\mathcal{M}_k \triangleq \left\{ (v, w) :  
				\mathbb{E}_{\mathbf{H}_k} \Big[ \big|\mathbf{w}_w^{\textnormal{H}} \mathbf{H}_k \mathbf{v}_v \big|^2\Big] \ge \xi \right\},
			\end{equation}
			where $\xi$ is a predefined power threshold.
		\end{definition}
		\begin{remark}
			The set $\mathcal{M}_k$ is solely dependent on the second order statistics of the channel $\mathbf{H}_k$, for fixed 
			$\mathcal{B}_{\text{BS}}$ and $\mathcal{B}_{\text{UE}}$. In particular, we refer to the notion of \emph{beam coherence time} to denote 
			the coherence time of such statistics. The beam coherence time $T_{\text{beam}}$ -- which depends on the beam width, \ac{UE} speed
			and other factors -- is much longer than the channel coherence time $T_{\text{coh}}$~\cite{Va2017}. \qed 
		\end{remark}
		The following lemma establishes the mathematical relation between the relevant channel components and the second order statistics of the 
		channel (channel covariance matrix). % which will be useful for algorithm derivation in Section~\ref{sec:Dec_Coo_BS_Algos}.
		\begin{lemma} \label{lemma:rel_ch_comp}
			Let $\mathbf{\Sigma}_k \triangleq \mathbb{E}_{\mathbf{H}_k} 
			\left[ \textnormal{vec}\left( \mathbf{H}_k \right)\textnormal{vec}\left( \mathbf{H}_k \right)^{\textnormal{H}}\right] \in
			\mathbb{C}^{N_{\textnormal{BS}} N_{\textnormal{UE}} \times N_{\textnormal{BS}} N_{\textnormal{UE}}}$ be the channel covariance matrix
			relative to the $k$-th UE. The set $\mathcal{M}_k$ containing the relevant channel components can be equivalently expressed as
			\begin{equation} \label{eq:lemma_rel_ch_comp}
				\mathcal{M}_k = \left\{ (v, w) :  \mathbf{b}_{v, w}^{\textnormal{H}} \mathbf{\Sigma}_k \mathbf{b}_{v, w} \ge \xi \right\},
			\end{equation}
			where $\mathbf{b}_{v, w} \triangleq \left( \textnormal{conj}\left(\mathbf{v}_v\right) \otimes \mathbf{w}_w \right)
			\in \mathbb{C}^{N_{\textnormal{BS}} N_{\textnormal{UE}} \times 1}$.
		\end{lemma}
		\begin{proof}
			From the properties of the $\text{vec}\left(\cdot\right)$ operator, we can write: 
			$\text{vec}\left( \mathbf{w}^{\text{H}} \mathbf{H} \mathbf{v} \right)
			= \left( \mathbf{v}^{\text{T}} \otimes \mathbf{w}^{\text{H}} \right) \text{vec}\left(\mathbf{H}\right)$. 
			Now, from the definition in \eqref{eq:rel_channel_comp}, we have:
			\begin{align}
				\mathbb{E}_{\mathbf{H}_k} \Big[ \big|\mathbf{w}_w^{\text{H}} \mathbf{H}_k \mathbf{v}_v \big|^2\Big]
				&= \mathbb{E}_{\mathbf{H}_k} \Big[ \text{vec}\left(\mathbf{H}\right)^\text{H} \mathbf{b}_{v, w}
				\mathbf{b}_{v, w}^{\text{H}} \text{vec}\left(\mathbf{H}\right) \Big] \nonumber \\
				&= \mathbb{E}_{\mathbf{H}_k} \Big[ \text{Tr} \left( \mathbf{b}_{v, w} \mathbf{b}_{v, w}^\text{H} 
				\text{vec}\left(\mathbf{H}\right) \text{vec}\left(\mathbf{H}\right)^{\text{H}} \right) \Big] \nonumber \\
				&= \text{Tr} \left(\mathbf{b}_{v, w}^\text{H} 
				\mathbb{E}_{\mathbf{H}_k} \left[ \text{vec}\left(\mathbf{H}\right) \text{vec}\left(\mathbf{H}\right)^{\text{H}} \right] 
				 \mathbf{b}_{v, w} \right) \nonumber \\
				 &= \mathbf{b}_{v, w}^{\textnormal{H}} \mathbf{\Sigma}_k \mathbf{b}_{v, w},
			\end{align}
			as given in \eqref{eq:lemma_rel_ch_comp}.
		\end{proof}
		The relevant channel components relative to the $k$-th \ac{UE} can thus be found through linear search over $B_{\text{BS}} B_{\text{UE}}$
		elements, provided that the second order statistics of $\mathbf{H}_k$ are known.
		
		Note that when the \acp{UE} exploit multi-beam covariance shaping, the set of relevant channel components can be altered\footnote{With the
		exception of \emph{spatially-uncorrelated} channels, where the same gain is expected from all spatial directions.}.
		Indeed, applying some receive beams means focusing on specific relevant beam pairs and neglecting some others. 
		To this end, we define the subset $\mathcal{M}_k^{\text{BS}} \subseteq \mathcal{M}_k$ as follows. 
		\begin{definition}
			We define the set $\mathcal{M}_k^{\textnormal{BS}} \subseteq \mathcal{M}_k$ containing the relevant channel components (or, equivalently, 
			beam pairs) of the $k$-th \ac{UE}, when the $k$-th \ac{UE} adopts $\mathbf{W}_k$ as its receive \ac{GoB} combiner, as follows:
			\begin{equation} \label{eq:rel_channel_comp_w}
				\mathcal{M}_k^{\textnormal{BS}}(\mathbf{W}_k) \triangleq 
				\left\{(v, w) \in \mathcal{M}_k : \mathbf{w}_w \in \mathbf{W}_k \right\},
			\end{equation}
			where we have introduced the notation $\mathcal{M}_k^{\text{BS}}(\cdot)$ to highlight that the set $\mathcal{M}_k^{\text{BS}}$ 
			depends on the selected \ac{GoB} combiner $\mathbf{W}_k$.
		\end{definition}
		
		In more detail, for given \ac{GoB} beamformers $\mathbf{V}_k \in \mathbb{C}^{N_{\text{BS}} \times M_{\text{BS}}}$ and 
		$\mathbf{W}_k \in \mathbb{C}^{N_{\text{UE}} \times M_{\text{UE}}}$, an effective channel covariance $\bar{\mathbf{\Sigma}}_k$ can be defined. 
		Furthermore, $\bar{\mathbf{\Sigma}}_k$ can be expressed in closed form as a function of the channel covariance 
		$\mathbf{\Sigma}_k$, as highlighted in the following lemma.
		\begin{lemma} \label{lemma:eff_ch_cov}
			Let $\bar{\mathbf{\Sigma}}_k \triangleq \mathbb{E}_{\mathbf{H}_k} 
			\left[ \textnormal{vec}\left( \mathbf{W}_k^{\textnormal{H}} \mathbf{H}_k \mathbf{V}_k \right) 
			\textnormal{vec}\left( \mathbf{W}_k^{\textnormal{H}} \mathbf{H}_k \mathbf{V}_k \right)^{\textnormal{H}}\right]
			\in \mathbb{C}^{M_{\textnormal{BS}} M_{\textnormal{UE}} \times M_{\textnormal{BS}} M_{\textnormal{UE}}}$ be the 
			effective channel covariance relative to the $k$-th \ac{UE}. $\bar{\mathbf{\Sigma}}_k$ can be equivalently expressed as
			\begin{equation} \label{eq:lemma_eff_ch_cov}
				\bar{\mathbf{\Sigma}}_k = \mathbf{B}_k^{\textnormal{H}} \mathbf{\Sigma}_k \mathbf{B}_k,
			\end{equation}
			where $\mathbf{B}_k \triangleq \left( \textnormal{conj}\left(\mathbf{V}_k\right) \otimes \mathbf{W}_k \right)
			\in \mathbb{C}^{N_{\textnormal{BS}} N_{\textnormal{UE}} \times M_{\textnormal{BS}} M_{\textnormal{UE}}}$.
		\end{lemma}
		\begin{proof}
			Based on $\text{vec}\left( \mathbf{W}^{\text{H}} \mathbf{H} \mathbf{V} \right)
			= \left( \mathbf{V}^{\text{T}} \otimes \mathbf{W}^{\text{H}} \right) \text{vec}\left(\mathbf{H}\right)$, we have:
			\begin{align}
				\bar{\mathbf{\Sigma}}_k &\triangleq \mathbb{E}_{\mathbf{H}_k} 
				\left[ \textnormal{vec}\left( \mathbf{W}_k^{\textnormal{H}} \mathbf{H}_k \mathbf{V}_k \right) 
				\textnormal{vec}\left( \mathbf{W}_k^{\textnormal{H}} \mathbf{H}_k \mathbf{V}_k \right)^{\textnormal{H}}\right] \nonumber \\
				&= \mathbf{B}_k^{\textnormal{H}} \mathbb{E}_{\mathbf{H}_k} 
				\Big[ \text{vec}\left(\mathbf{H}_k\right)^\text{H} \text{vec}\left( \mathbf{H}_k \right) \Big] \mathbf{B}_k \nonumber \\
				 &= \mathbf{B}_k^{\textnormal{H}} \mathbf{\Sigma}_k \mathbf{B}_k,
			\end{align}
			as given in \eqref{eq:lemma_eff_ch_cov}.
		\end{proof}
		
		Let us now consider the single-user optimal \ac{SVD} precoding~\cite{Tse2005} over the effective channels. 
		We can express the achievable \ac{SE} at the $k$-th \ac{UE} as follows:
		\begin{equation}
			\mathcal{R}^{\text{SVD}}_k \left( \mathbf{V}_k, \mathbf{W}_k \right) \triangleq
			\log_2 \det  \left( \mathbf{I}_{M_{\text{UE}}} + \kappa \mathbf{\Lambda}^{\text{H}}_k \mathbf{\Lambda}_k \right),
		\end{equation}
		where we recall that $\kappa \triangleq \rho^2 / \sigma_n^2$ and where
		$\mathbf{\Lambda} \triangleq \text{diag}\left(\lambda_1, \dots, \lambda_{M_{\text{UE}}} \right)$, with
		$\lambda_1, \dots, \lambda_{M_{\text{UE}}}$ being the singular values of the effective channel 
		$\bar{\mathbf{H}}_k \triangleq \mathbf{W}^{\text{H}}_k \mathbf{H}_k \mathbf{V}_k$.
		
		\begin{proposition} \label{prop:se_svd_approx}
			The average \ac{SE} achievable at the $k$-th \ac{UE} in a single-user scenario with \ac{SVD} precoding can be upper bounded as follows:
			\begin{equation} \label{eq:se_svd_approx}
				\mathbb{E}_{\mathbf{H}_k} \Big[ \mathcal{R}^{\textnormal{SVD}}_k \left( \mathbf{V}_k, \mathbf{W}_k \right) \Big] \le
				M_{\textnormal{UE}} \log_2 \left( 1 + \kappa M_{\textnormal{UE}}^{-1} \textnormal{Tr}\left(\bar{\mathbf{\Sigma}}_k\right)  \right),	
			\end{equation}
			where $\bar{\mathbf{\Sigma}}_k$ is the effective channel covariance relative to the $k$-th \ac{UE}.
		\end{proposition}
		\begin{proof}
			We first rewrite $\mathcal{R}^{\text{SVD}}_k \left( \mathbf{V}_k, \mathbf{W}_k \right)$ as follows, using the properties of the 
			$\det$ operator:
			\begin{equation}
				\mathcal{R}^{\text{SVD}}_k \left( \mathbf{V}_k, \mathbf{W}_k \right) = 
				\sum_{m=1}^{M_{\text{UE}}} \log_2 \left( 1 + \kappa \lambda_m^2 \right).
			\end{equation}
			According to \emph{Jensen's inequality}, we have:
			\begin{align}
				\sum_{m=1}^{M_{\text{UE}}} \log_2 \left( 1 + \kappa \lambda_m^2 \right) &\le
				M_{\text{UE}} \log_2 \left( 1 + \kappa M_{\text{UE}}^{-1} \sum_{m=1}^{M_{\text{UE}}} \lambda_m^2 \right) \nonumber \\
				&= M_{\text{UE}} \log_2 \left( 1 + \kappa M_{\text{UE}}^{-1} \text{Tr}\left( \bar{\mathbf{H}}_k \bar{\mathbf{H}}^{\text{H}}_k 
				\right)\right).
			\end{align}
			Now, considering the expectation and again exploiting \emph{Jensen's inequality}, we can write:
			\begin{align}
				\mathbb{E}_{\mathbf{H}_k} \Big[ 
				M_{\text{UE}} \log_2 \left( 1 + \kappa M_{\text{UE}}^{-1} \text{Tr}\left( \bar{\mathbf{H}}_k \bar{\mathbf{H}}^{\text{H}}_k \right)\right) 
				\Big] &\le M_{\text{UE}} \log_2 \left( 1 + \kappa M_{\text{UE}}^{-1} \mathbb{E}_{\mathbf{H}_k} \Big[ 
				\text{Tr}\left( \bar{\mathbf{H}}_k \bar{\mathbf{H}}^{\text{H}}_k \right) \Big] \right) \nonumber \\
				&= M_{\text{UE}} \log_2 \left( 1 + \kappa M_{\text{UE}}^{-1} \mathbb{E}_{\mathbf{H}_k} \Big[ 
				\text{Tr}\left( \text{vec}\left(\bar{\mathbf{H}}_k\right) \text{vec}\left( \bar{\mathbf{H}}_k \right)^{\text{H}} \right) \Big] 
				\right) \nonumber \\
				&=M_{\text{UE}} \log_2 \left( 1 + \kappa M_{\text{UE}}^{-1} \text{Tr}\left(\bar{\mathbf{\Sigma}}_k\right)  \right),
			\end{align}
			as given in \eqref{eq:se_svd_approx}.
		\end{proof}
		\begin{corollary}
			The upper bound of the average \ac{SE} in \eqref{eq:se_svd_approx} is maximized when the effective channel covariance 
			$\bar{\mathbf{\Sigma}}_k$ is shaped through a beam selection based on the relevant beams. \qed
		\end{corollary}
		\begin{proof}
			From Lemma~\ref{lemma:eff_ch_cov} we have:
			\begin{align} \label{eq:sum_trace_rel_ch_comp}
				\text{Tr} \left( \bar{\mathbf{\Sigma}}_k \right) 
				&= \text{Tr} \left( \mathbf{B}_k^{\textnormal{H}} \mathbf{\Sigma}_k \mathbf{B}_k \right) \nonumber \\
				&= \sum_{m=1}^{M_{\text{BS}} M_{\text{UE}}} \mathbf{b}_m^{\text{H}} \mathbf{\Sigma}_k \mathbf{b}_m.
			\end{align}			
			The sum in \eqref{eq:sum_trace_rel_ch_comp} is maximized when the first $M_{\text{BS}} M_{\text{UE}}$ strongest channel components 
			are selected for transmission, among the ones in $\mathcal{M}_k$.
		\end{proof}
		Fig.~\ref{fig:avg_se_jensen} shows that \eqref{eq:se_svd_approx} is a good upper-bound of the actual average \ac{SE}, that can be used to 
		perform beam selection. Thus, as a first approximation towards the maximization of the overall effective network throughput as 
		in \eqref{eq:gen_opt_after_BD}, we formulate the uncoordinated beam selection problem \eqref{eq:un_beam_sel} which aims to maximize 
		instead the sum \ac{SE} defined as $\sum_{k=1}^K \mathcal{R}^{\text{SVD}}_k \left( \mathbf{V}_k, \mathbf{W}_k \right)$:
		\begin{align} \label{eq:un_beam_sel} \tag{P1}
			\left(\mathbf{V}^{\text{(P1)}}, \mathbf{W}^{\text{(P1)}}\right) = 
			& \argmax_{\mathbf{V}, \mathbf{W}} ~\sum_{k=1}^K 
			M_{\textnormal{UE}} \log_2 \left( 1 + \kappa M_{\textnormal{UE}}^{-1} \textnormal{Tr}\left(\bar{\mathbf{\Sigma}}_k\right)  \right), \\
			& \text{subject to} ~\text{col}\left(\mathbf{V}\right) \in \mathcal{B}_{\text{BS}} \nonumber \\
			& \phantom{\text{subject to}} ~\text{col}\left(\mathbf{W}_k\right) \in \mathcal{B}_{\text{UE}}, ~\forall k = \{ 1, \dots, K \}. \nonumber
		\end{align}
		Since the objective function in \eqref{eq:un_beam_sel} is \emph{disjoint} with the \acp{UE}, 
		\eqref{eq:un_beam_sel} can be solved through letting each \ac{UE} maximizing its own related term in the sum. 
		In particular, the $k$-th \ac{UE} shapes its channel covariance using the beams in $\mathcal{M}_k$. 
		Such a task requires a linear search over the $B_{\text{BS}} B_{\text{UE}}$ elements in the \ac{GoB} codebooks.
		The precoder matrix at the \ac{BS} side (common to all the \acp{UE}) is then constructed as 
		$\text{col}\left(\mathbf{V}\right) = \cup_{k=1}^K \text{col}\left(\mathbf{V}_k\right)$.
		The relevant channel components (or beams) offers thus a straightforward method to design the \ac{GoB} beamformers.
		Note that \eqref{eq:un_beam_sel} represents the baseline performance for the algorithms that we will introduce in the following.
		Uncoordinated \ac{SNR}-based approaches like \eqref{eq:un_beam_sel} can be found in previous works related to \emph{hybrid beamforming}, 
		which apply to \ac{FDD} \ac{mMIMO} as well. For example, the authors in~\cite{Kim2013} focus on single-user beam selection and 
		propose strategies based on the received signal strength. Likewise, in~\cite{Alkhateeb2015}, a multi-user beam selection method is proposed,
		where the analog beams are chosen according to the strongest paths at each \ac{UE}. Although our derivation of \eqref{eq:un_beam_sel} falls 
		within covariance shaping and differs from previous works, the resulting beam selection can be considered de facto identical.
		
		\begin{figure}[h]
			\centering
			\includegraphics[trim=1.57in 3.37in 1.57in 3.57in, clip, scale=0.687]{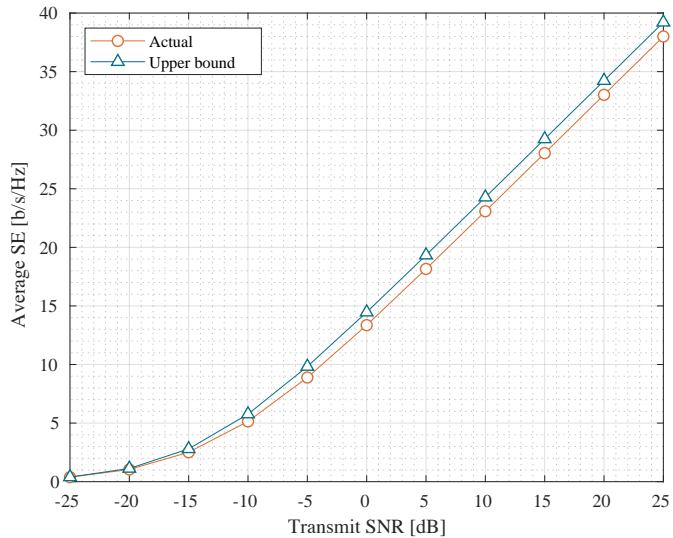}
			\caption{Average \ac{SE} vs SNR for a single-user case where beam selection is based on the relevant beams. 
			The simulation parameters are as follows: $N_{\text{BS}} = 64$, $N_{\text{UE}} = 4$, $M_{\text{BS}} = 5$, $M_{\text{UE}} = 3$. 
			The upper bound obtained in \eqref{eq:se_svd_approx} can be used as a tight approximation for the actual \ac{SE}.}
			\label{fig:avg_se_jensen}
		\end{figure}
				
	\subsection{Minimizing Multi-User Interference} \label{sec:min_mu_int}
		
		As well-captured in Fig.~\ref{fig:scen}, the uncoordinated selection of the \ac{GoB} beamformers as in \eqref{eq:un_beam_sel} can lead to
		overall inefficient strategies in terms of training overhead and multi-user interference reduction. As opposed to uncoordinated approaches, 
		clever coordinated beam selection strategies can help shaping the effective channel subspaces so as to optimize the multi-user transmission. 
		In this section, we will show that a proper beam selection can be made so as to take multi-user interference into account within the covariance
		shaping process. % To this end, we will introduce the so-called \acf{GCMD}.

		As seen in Section~\ref{sec:Data_BF_Design}, the \ac{BD} approach imposes two crucial conditions on the overall effective channel 
		$\bar{\mathbf{H}} \triangleq \mathbf{W}^{\text{H}} \mathbf{H} \mathbf{V}$ for transmitting data without multi-user interference:
		\begin{itemize}
			\item{No inter-user interference $\iff \left|\text{null}\left(\bar{\mathbf{H}}_{/k}\right)\right| \neq 0$;}
			\item{No inter-stream interference $\iff \text{rank}\left(\bar{\mathbf{H}}_k \bar{\mathbf{M}}_{/k}^{(0)}\right) \ge 1$.}
		\end{itemize}
		From Lemma~\ref{lemma:interference_cancel}, we know that the second condition requires at least one vector in 
		$\text{row}\left(\bar{\mathbf{H}}_k\right)$ that is \ac{LI} of $\text{row}\left(\bar{\mathbf{H}}_{/k}\right)$.
		% In this case, it is possible to find a precoding vector which removes the interference coming from the other simultaneous transmissions.
		Opposite to \ac{TDD} \ac{mMIMO}, where the \ac{LMMSE} can return \ac{LI} channel estimates depending on the propagation
		environment~\cite{Bjornson2018}, the estimates in \eqref{eq:mmse_estimate} are \ac{LI} \emph{almost surely}, due to independent channel
		realizations and estimation processes at the \ac{UE} side. Therefore, the second condition for multi-user interference cancellation is 
		\emph{always} respected in the case of downlink training with \ac{LMMSE} estimation at the \ac{UE} side.
		
		On the other hand, whenever $(K-1) M_{\text{UE}} < M_{\text{BS}}$, then 
		$\left|\text{null}\left(\bar{\mathbf{H}}_{/k}\right)\right| \ge M_{\text{BS}} - (K-1) M_{\text{UE}} > 0$. 
		Therefore, in such a case, it is \emph{always} possible (for whatever $\mathbf{V}$ and $\mathbf{W}$) to find a matrix 
		$\bar{\mathbf{M}}_{/k}^{(0)}$ in \eqref{eq:svd_h_no_k} different from the null matrix $\mathbf{0}$ and as such, to remove 
		multi-user interference. In this case, the minimum training overhead becomes proportional to $K$ and comparable to the one needed 
		in \ac{TDD} operation (although in \ac{TDD} such overhead is generated in the uplink channel). The bottom line is that 
		$\left(K-1\right)M_{\text{UE}} < M_{\text{BS}}$ is the \emph{only} condition that the \ac{BD} precoding imposes on the \ac{GoB} 
		beamformers design in order to suppress multi-user interference.
		
		Nevertheless, the \ac{BD} precoding affects the received gain at the generic $k$-th \ac{UE}. In particular, depending on how much the effective
		channels in $\bar{\mathbf{H}}$ are \emph{spatially-separated}, the application of the precoding matrix $\bar{\mathbf{M}}_{/k}^{(0)}$ on 
		$\bar{\mathbf{H}}_k$ can lead to a drastic gain loss compared to the single-user case (refer to Proposition~\ref{prop:se_svd_approx}). 
		In order to infer such loss, the so-called \ac{CMD} can be used. 
		% The \ac{CMD} has been introduced in~\cite{Herdin2005} to measure the variation of the second-order statistics for fast-moving \acp{UE}.
		The \ac{CMD} has been exploited in~\cite{Mursia2018} to increase the \emph{spatial separability} among 
		the \acp{UE} through covariance shaping at the \ac{UE} side. The authors in~\cite{Mursia2018} consider a $2$-\ac{UE} case. 
		For multiple \acp{UE}, we introduce the \acf{GCMD} as follows.
		\begin{definition}
			We define the \ac{GCMD} $\delta_k\left(\mathbf{\Sigma}_1, \dots, \mathbf{\Sigma}_K\right) \in [0, 1]$ between the channel covariance 
			$\mathbf{\Sigma}_k$ of the $k$-th \ac{UE} and the channel covariance
			$\mathbf{\Sigma}_j$ of the $j$-th \ac{UE}, where $j \in \{1, \dots, K\} \backslash \{k\}$ as
			\begin{equation}
				\delta_k\left(\mathbf{\Sigma}_1, \dots, \mathbf{\Sigma}_K\right) \triangleq 
				1 - \frac{1}{K-1} \sum_{\substack{j=1 \\ j \neq k}}^K \frac{\textnormal{Tr}\left( \mathbf{\Sigma}_k \mathbf{\Sigma}_j \right)}
				{\norm{\mathbf{\Sigma}_k}_{\textnormal{F}} \norm{\mathbf{\Sigma}_j}_{\textnormal{F}}}.
			\end{equation}
		\end{definition}
		Note that the \emph{spatial orthogonality} condition, i.e. $\textnormal{Tr}\left( \mathbf{\Sigma}_k \mathbf{\Sigma}_j \right) = 0,
		~\forall j \neq k$, which was exploited in several other studies related to \ac{FDD} \ac{mMIMO} optimization~\cite{Yin2013, Adhikary2013} is
		equivalent to $\delta_k\left(\mathbf{\Sigma}_1, \dots, \mathbf{\Sigma}_K\right) = 1, ~\forall k$. This is a desirable spatial condition for 
		which the \ac{BD} incurs no reduction of the channel gain. On the other hand, the \ac{GCMD} becomes zero when the covariance matrices of 
		the \acp{UE} are equal up to a scaling factor. Both these extreme conditions are seldom experienced in practical 
		scenarios~\cite{Gao2015, Bjornson2018}. Nevertheless, when the channel covariances are shaped through statistical beamforming, 
		resulting in some \emph{effective} channel covariances, the \ac{GCMD} can be used as a metric to evaluate how the covariance shaping 
		affects the \emph{spatial separability} of the \acp{UE}. 
		
		In particular, we use the \ac{GCMD} evaluated on the effective covariances 
		$\bar{\mathbf{\Sigma}}_k \triangleq \mathbf{B}_k^{\textnormal{H}} \mathbf{\Sigma}_k \mathbf{B}_k, ~\forall k$, 
		where $\mathbf{B}_k \triangleq \left( \textnormal{conj}\left(\mathbf{V}\right) 
		\otimes \mathbf{W}_k \right)$, to introduce a \emph{penalty factor} in $\mathcal{R}^{\text{SVD}}_k \left( \mathbf{V}_k, \mathbf{W}_k \right)$,
		so as to approximate the \ac{SE} in \eqref{eq:SINR_gen_BD} achieved after \ac{BD} precoding. 
		In this case, the \ac{GoB} beamformers $\left(\mathbf{V}^{\text{(P2)}}, \mathbf{W}^{\text{(P2)}}\right)$ are
		obtained through solving the coordinated beam selection problem \eqref{eq:cov_beam_sel}, as follows:
		\begin{align} \label{eq:cov_beam_sel} \tag{P2}
			\left(\mathbf{V}^{\text{(P2)}}, \mathbf{W}^{\text{(P2)}}\right) = 
			&\argmax_{\mathbf{V}, \mathbf{W}} ~\sum_{k=1}^K 
			M_{\textnormal{UE}} \log_2 \left( 1 + \kappa M_{\textnormal{UE}}^{-1} \textnormal{Tr}\left(\bar{\mathbf{\Sigma}}_k\right) 
			\delta_k\left(\bar{\mathbf{\Sigma}}_1, \dots, \bar{\mathbf{\Sigma}}_K\right) \right), \\
			& \text{subject to} ~\left(K-1\right)M_{\text{UE}} < M_{\text{BS}} \nonumber \\
			& \phantom{\text{subject to}} ~\text{col}\left(\mathbf{V}\right) \in \mathcal{B}_{\text{BS}} \nonumber \\
			& \phantom{\text{subject to}} ~\text{col}\left(\mathbf{W}_k\right) \in \mathcal{B}_{\text{UE}}, ~\forall k = \{ 1, \dots, K \}. \nonumber
		\end{align}
		In \eqref{eq:cov_beam_sel}, the beam decision at the $k$-th \ac{UE} influences the beam decisions at all the other
		\acp{UE}. Therefore, a central coordinator knowing all the large-dimensional channel covariances 
		$\mathbf{\Sigma}_k, ~\forall k$ and which dictates the beam strategies to each \ac{UE} is needed to solve this
		problem. In Section~\ref{sec:Dec_Coo_BS_Algos}, we will propose a hierarchical approach to circumvent this issue.
		Note that in both \eqref{eq:un_beam_sel} and \eqref{eq:cov_beam_sel} we have proposed approximations 
		of the \ac{SE} which neglect the pre-log factor relative to the training overhead. We will now look into the third condition required for an 
		effective \ac{GoB} beamformers design in the \ac{FDD} \ac{mMIMO} regime, which is the minimization of the training overhead.

	\subsection{Minimizing Training Overhead} \label{sec:ue_rep_train_ov}
		
		As seen in Section~\ref{sec:Coo_Beam_Rep}, there is a direct relation between the training overhead and the design of the \ac{GoB} precoder 
		$\mathbf{V}$ (refer to Definition~\ref{def:train_ov_v}). In particular, under the \ac{GoB} assumption, the training overhead 
		$\omega\left(\mathbf{V}\right)$ ranges in $\left[1, (\tau/T) B_{\text{BS}} \right]$, where the right extreme is experienced when all the beams 
		in the codebook $\mathcal{B}_{\text{BS}}$ are trained. The more beams are trained, the more spatial degrees of freedom are obtained (higher
		spatial multiplexing and beamforming gain). However, when considering the training overhead, adding more and more beams is \emph{likely} 
		to result in diminishing returns~\cite{Zirwas2016}. 
		
		The alternative is to train the relevant channel components~\cite{Zirwas2016, Xiong2017}, as those relate to the spatial subspaces which give 
		the strongest gain. In this case, the design of the precoder $\mathbf{V}$ is not separated from the design of the combiners $\mathbf{W}$, as well
		captured in \eqref{eq:rel_channel_comp_w}. In the current $3$GPP specifications, a beam reporting procedure is designed to assist the \ac{BS} in the
		precoder selection~\cite{3GPP2018c}. Such procedure has a direct impact on the performance of the \ac{DL} \ac{SE}. In particular, the $k$-th
		\ac{UE} reports to the \ac{BS} the set $\mathcal{M}_k^{\text{BS}}(\mathbf{W}_k)$ -- also known as \ac{PMI}~\cite{3GPP2018c} -- following an
		appropriate \ac{GoB} combiner (beam) selection. To this end, we reformulate the definition of the training overhead, depending on the beam
		decisions carried out at the \ac{UE} side.
		\begin{definition}
			Let $\mathbf{W} \in \mathbb{C}^{K N_{\textnormal{UE}} \times K M_{\textnormal{UE}}}$ be the overall \ac{GoB} combiner as in
			\eqref{eq:overall_GoB_combiner}. The training overhead $\omega\left(\mathbf{W}\right)$ is defined as follows:
			\begin{equation} \label{eq:training_overhead_w}
				\omega\left(\mathbf{W}\right) \triangleq \frac{\tau}{T} 
				\textnormal{card}\left(\bigcup_{k=1}^K \mathcal{M}_k^{\textnormal{BS}}(\mathbf{W}_k)\right).
			\end{equation}
		\end{definition}
		In the $3$GPP implementation, the beam decisions carried out at each \ac{UE} have thus a central role in affecting the training overhead 
		under the \ac{GoB} approach. Note that $\omega\left(\mathbf{W}\right)$ can increase and approach the extreme value $(\tau/T) B_{\text{BS}}$ 
		in heterogeneous propagation environments with rich scattering, due to the growing number of relevant beams to activate at the \ac{BS} 
		side~\cite{Zirwas2016}. In this respect, adopting approaches such as \eqref{eq:un_beam_sel} or \eqref{eq:cov_beam_sel} 
		for selecting the beams can undermine the application of the \ac{GoB} approach in multi-user scenarios. 
		
		On the other hand, the largest training overhead reduction is achieved when the \acp{UE} coordinate in the beam domain so that
		to activate the smallest possible number of beams at the BS.
		In general, a balance between achievable beamforming gain and required training overhead, as well as multi-user interference, has to be 
		considered in the beam decision process and combiner selection at the \acp{UE}. In the following, we formulate two optimization problems which 
		take the pre-log factor relative to the training overhead into account. In the first one, the pre-log term is added in the objective function of the
		optimization problem \eqref{eq:un_beam_sel}. Thus, we introduce the coordinated beam selection problem \eqref{eq:coo_beam_sel}, where 
		both the achieved channel gain and the training overhead are taken into account, as follows:
		\begin{align} \label{eq:coo_beam_sel} \tag{P3}
			\left(\mathbf{V}^{\text{(P3)}}, \mathbf{W}^{\text{(P3)}}\right) = 
			&\argmax_{\mathbf{V}, \mathbf{W}} ~(1-\omega\left(\mathbf{W}\right)) \sum_{k=1}^K 
			M_{\textnormal{UE}} \log_2 \left( 1 + \kappa M_{\textnormal{UE}}^{-1} \textnormal{Tr}\left(\bar{\mathbf{\Sigma}}_k\right) \right), \\
			& \text{subject to} ~\left(K-1\right)M_{\text{UE}} < M_{\text{BS}} \nonumber \\
			& \phantom{\text{subject to}} ~\text{col}\left(\mathbf{V}\right) \in \mathcal{B}_{\text{BS}} \nonumber \\
			& \phantom{\text{subject to}} ~\text{col}\left(\mathbf{W}_k\right) \in \mathcal{B}_{\text{UE}}, ~\forall k = \{ 1, \dots, K \}, \nonumber
		\end{align}
		where we recall that $\left(K-1\right)M_{\text{UE}} < M_{\text{BS}}$ relates to the crucial condition that the \ac{BD} 
		precoding imposes on the \ac{GoB} beamformers design in order to suppress multi-user interference.
		
		The last optimization problem that we introduce aims at balancing the three conditions for an effective \ac{GoB} beamformers design that we
		have considered in this section. As such, the long-term beam selection problem \eqref{eq:csh_beam_sel} includes the pre-log factor relative
		to the training overhead in the objective function of the problem \eqref{eq:cov_beam_sel}. The optimum \ac{GoB} beamformers
		$\left(\mathbf{V}^{\text{(P4)}}, \mathbf{W}^{\text{(P4)}}\right)$ are thus obtained through solving the following optimization problem:
		\begin{align} \label{eq:csh_beam_sel} \tag{P4}
			\left(\mathbf{V}^{\text{(P4)}}, \mathbf{W}^{\text{(P4)}}\right) = 
			&\argmax_{\mathbf{V}, \mathbf{W}} ~(1-\omega\left(\mathbf{W}\right)) \sum_{k=1}^K 
			M_{\textnormal{UE}} \log_2 \left( 1 + \kappa M_{\textnormal{UE}}^{-1} \textnormal{Tr}\left(\bar{\mathbf{\Sigma}}_k\right) 
			\delta_k\left(\bar{\mathbf{\Sigma}}_1, \dots, \bar{\mathbf{\Sigma}}_K\right) \right), \\
			& \text{subject to} ~\left(K-1\right)M_{\text{UE}} < M_{\text{BS}} \nonumber \\
			& \phantom{\text{subject to}} ~\text{col}\left(\mathbf{V}\right) \in \mathcal{B}_{\text{BS}} \nonumber \\
			& \phantom{\text{subject to}} ~\text{col}\left(\mathbf{W}_k\right) \in \mathcal{B}_{\text{UE}}, ~\forall k = \{ 1, \dots, K \}. \nonumber
		\end{align}
		The same conclusions drawn for the optimization problem \eqref{eq:cov_beam_sel} are valid for \eqref{eq:coo_beam_sel} and \eqref{eq:csh_beam_sel}.
		In particular, to solve \eqref{eq:csh_beam_sel}, the central coordinator needs to know the \acp{PMI} 
		$\mathcal{M}_k^{\text{BS}}\left( \mathbf{W}_k \right), ~\forall k$ in addition to the channel covariances $\mathbf{\Sigma}_k, ~\forall k$.
		
		Fig.~\ref{fig:opt_probs_approxs} compares the effective network throughput $\mathcal{R}$ as in \eqref{eq:gen_opt_after_BD} with its 
		approximations in \eqref{eq:un_beam_sel}-\eqref{eq:csh_beam_sel}. 
		The approximated objective function of the optimization problem \eqref{eq:csh_beam_sel} gives the tightest 
		upper bound to the actual effective network throughput as expected. We summarize the proposed
		optimization problems and their considered sub-problems as introduced above in Table~\ref{tab:problems}. 
		
		In the next section, we will propose a framework exploiting \ac{D2D} communications which will allow for a decentralized implementation 
		of a series of beam selection algorithms based on the problems \eqref{eq:cov_beam_sel}-\eqref{eq:csh_beam_sel} described above. 
		The nature of such problems is such that \eqref{eq:un_beam_sel}-\eqref{eq:csh_beam_sel} offer and explore various 
		\emph{complexity-performance} trade-offs interesting from the implementation perspective.
		
		\setlength{\belowcaptionskip}{-1em}
		\begin{table}[h]
			\centering
		  	\caption{The proposed optimization problems \eqref{eq:un_beam_sel}-\eqref{eq:csh_beam_sel} with their considered sub-problems.}
	    		\label{tab:problems}
	    		\begin{tabular}{|c|c|c|c|} \hline
	      			\textbf{Problem} & \emph{$\max$ channel gain} & \emph{$\min$ multi-user interference} & \emph{$\min$ training overhead} \\
	      			\hline
	      			Uncoordinated (P1) & \Checkedbox & \HollowBox & \HollowBox \\
	      			Coordinated (P2) & \Checkedbox & \Checkedbox & \HollowBox \\ 
	      			Coordinated (P3) & \Checkedbox & \HollowBox & \Checkedbox \\ 
	      			Coordinated (P4) & \Checkedbox & \Checkedbox & \Checkedbox \\
	      			\hline
	    		\end{tabular}
		\end{table}
		\begin{figure}[h]
			\centering
			\includegraphics[trim=1.57in 3.37in 1.57in 3.57in, clip, scale=0.78]{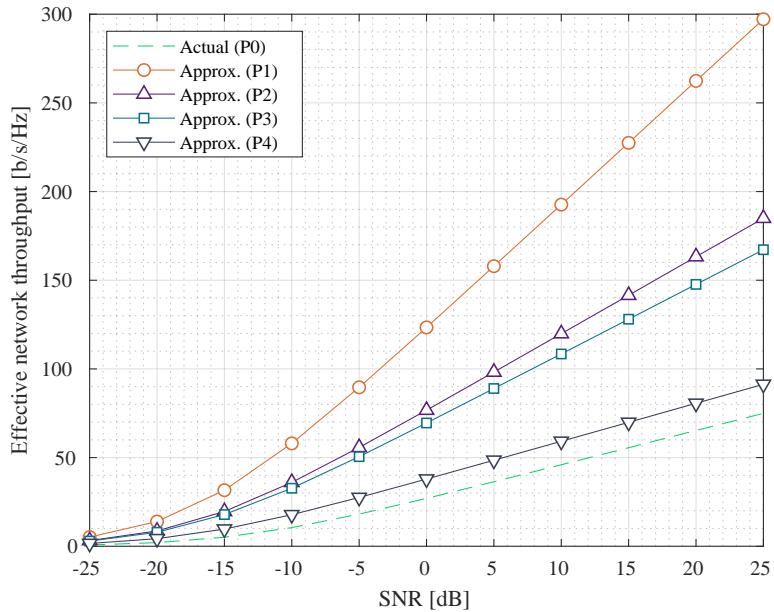}
			\caption{Comparison of the actual effective network throughput $\mathcal{R}$ as in \eqref{eq:ent} and its approximations defined in
			\eqref{eq:gen_opt_after_BD}-\eqref{eq:csh_beam_sel}. $K = 7$ \acp{UE}. 
			% The beam selection at each \ac{UE} is based on the local \ac{SNR}. 
			The approximation \eqref{eq:csh_beam_sel} is the closest to the actual throughput.}
			\label{fig:opt_probs_approxs}
		\end{figure}
						
\section{Decentralized Coordinated Beam Selection Algorithms} \label{sec:Dec_Coo_BS_Algos}
	
	Although no instantaneous information is needed to solve \eqref{eq:cov_beam_sel}-\eqref{eq:csh_beam_sel},
	such problems still require a central coordinator that knows the channel covariances 
	$\mathbf{\Sigma}_k ~\forall k$ and dictates the beam strategies to each \ac{UE}. 
	This is because the beam decisions at the generic $k$-th \ac{UE} affect the beam decisions at all the other \acp{UE}. 
	In this respect, collecting such large-dimensional statistical information at a central coordinator 
	as e.g. the \ac{BS} involves additional resource overhead. %~\cite{Love2008}.
	
	In order to achieve decentralized coordination, we propose to use a hierarchical information structure requiring small overhead. 
	In particular, an \emph{(arbitrary)} order among the UEs is established\footnote{The hierarchical order of the \acp{UE} has 
	a clear impact on the performance of the proposed scheme. In this paper, we will consider a random order and leave aside 
	further analysis on how such \emph{hierarchy} is defined and maintained.}, for which the $k$-th UE has access to some long-term 
	statistical information of the (lower-ranked) \acp{UE} $1, \dots, k-1$.
	This configuration is obtainable through e.g. \ac{D2D} communications. In this respect, the $3$GPP Release $16$ is expected to support point-to-point
	side-links which facilitate cooperative communications among the \acp{UE} with low resource consumption~\cite{Fodor2016, 3GPP2018d}. 
	The \emph{recently-specified} \ac{NR} side-link is thus a cornerstone for the proposed scheme. We further assume that such exchanged information
	is \emph{perfectly} decoded at the intended UEs.
	
	The full signaling sequence of the proposed hierarchical beam selection is given in Fig.~\ref{fig:signaling}. The core part of the procedure resides 
	in the long-term beam decision made at each UE on a beam coherence time basis so that the respective objective function is maximized.
	Based on the objective functions in \eqref{eq:cov_beam_sel}-\eqref{eq:csh_beam_sel}, we consider $3$ different beam decision policies, as in
	\eqref{eq:obj_func_dec}. Such policies have different requirements concerning the statistical information to exchange through \ac{D2D} side-links.
	In Table~\ref{tab:algorithms}, we summarize the differences between the proposed beam selection policies based on 
	\eqref{eq:un_beam_sel}-\eqref{eq:csh_beam_sel} with respect to the required information at the $k$-th \ac{UE}.
	
	\begin{table}[h]
		\centering
		\caption{The proposed algorithms and their required information at the $k$-th \ac{UE}.
  		The information relative to the (lower-ranked) \acp{UE} $1, \dots, k-1$ is exchanged through \ac{D2D} side-links.}
    	\label{tab:algorithms}
    	\begin{tabular}{|c|c|c|} \hline
      		\textbf{Algorithm} & \emph{Required local information} & \emph{Required information to be exchanged through \ac{D2D}} \\
      		\hline
      		Uncoordinated (P1) & $\mathbf{\Sigma}_k$ & \emph{Nothing} \\
      		Coordinated (P2) & $\mathbf{\Sigma}_k$ & $\bar{\mathbf{\Sigma}}_j, ~j \in \{ 1, \dots, k-1 \}$ \\ 
      		Coordinated (P3) & $\mathbf{\Sigma}_k$ & $\mathcal{M}_j^{\text{BS}}(\mathbf{W}_j), ~j \in \{ 1, \dots, k-1 \}$ \\ 
      		Coordinated (P4) & $\mathbf{\Sigma}_k$ & $\bar{\mathbf{\Sigma}}_j, 
      		~\mathcal{M}_j^{\text{BS}}(\mathbf{W}_j), ~j \in \{ 1, \dots, k-1 \}$\\ \hline
    	\end{tabular}
	\end{table}
	\begin{figure}[h]
		\centering
		\resizebox{11.07cm}{!}{
			\begin{tikzpicture}
				% HIERARCHICAL BEAM SELECTION BLOCK RECTANGLE
				\filldraw[fill=EurecomBlue, opacity=0.27] (-5,0) rectangle (7.17,4.87);
				% VERTICAL LINES
				\draw (-3.5,5.5) -- (-3.5,-6);
				\draw (-.5,5.5) -- (-.5,-6);
				\draw (2.5,5.5) -- (2.5,-6);
				\draw (5.5,5.5) -- (5.5,-6);
				% END OF VERTICAL LINES (dots)
				\draw[thick, loosely dotted] (-3.5,-6) -- (-3.5,-6.5);
				\draw[thick, loosely dotted] (-.5,-6) -- (-.5,-6.5);
				\draw[thick, loosely dotted] (2.5,-6) -- (2.5,-6.5);
				\draw[thick, loosely dotted] (5.5,-6) -- (5.5,-6.5);
				% TOP RECTANGLES
				\filldraw[fill=white] (-4,5) rectangle (-3, 5.5) node[pos=.5] {\small{BS}};
				\filldraw[fill=white] (-1,5) rectangle (0, 5.5) node[pos=.5] {\small{UE $3$}};
				\filldraw[fill=white] (2,5) rectangle (3, 5.5) node[pos=.5] {\small{UE $2$}};
				\filldraw[fill=white] (5,5) rectangle (6, 5.5) node[pos=.5] {\small{UE $1$}};
				% STEPS
				\fill[pattern=north west lines, pattern color=EurecomBlueDark] 
				(-3.5,-2) rectangle (5.5,-1.5) node[pos=.5, fill=white]{\footnotesize Training phase};
				\fill[pattern=north west lines, pattern color=EurecomBlueDark] 
				(-3.5,-3.5) rectangle (5.5,-3) node[pos=.5, fill=white]{\footnotesize Training phase};
				\fill[pattern=north west lines, pattern color=EurecomBlueDark] 
				(-3.5,-5) rectangle (5.5,-4.5) node[pos=.5, fill=white]{\footnotesize Training phase};
				\filldraw[fill=white] (4.3,4) rectangle (6.7, 4.5) node[pos=.5] {\small Beam decision};
				\draw[<-, thick, orange] (2.5, 3.5) -- (5.5, 3.5) node[pos=.5, above]{\small D2D};
				\draw[<-, thick, orange] (-.5, 3.37) -- (5.5, 3.37)
				node[black, right]{\footnotesize $\mathcal{M}_1^{\text{BS}}\left(\mathbf{W}_1^* \right)$};
				\draw[<-, dashed] (-3.5, 3.24) -- (5.5, 3.24);
				\filldraw[fill=white] (1.3,2.5) rectangle (3.7, 3) node[pos=.5] {\small Beam decision};
				\draw[<-, thick, orange] (-.5, 2) -- (2.5, 2) node[pos=.5, above]{\small D2D}
				node[black, right]{\footnotesize $\mathcal{M}_2^{\text{BS}}\left(\mathbf{W}_2^* \right)$};
				\draw[<-, dashed] (-3.5, 1.87) -- (2.5, 1.87);
				\filldraw[fill=white] (-1.7,1) rectangle (.7, 1.5) node[pos=.5] {\small Beam decision};
				\draw[<-, dashed] (-3.5, .5) -- (-.5, .5)
				node[black, right]{\footnotesize $\mathcal{M}_3^{\text{BS}}\left(\mathbf{W}_3^* \right)$};
				\filldraw[white] (2, .5) rectangle (7,1) 
				node[pos=.5, black]{\textbf{\emph{\footnotesize Hierarchical beam selection block}}};
				\filldraw[fill=white] (-5.5, -1) rectangle (-1.5, -.5) node[pos=.5]{\footnotesize Select the GoB precoder};
				\draw[->, dashed] (-3.5, -1.5) node[left]{$\mathbf{VS}$} -- (5.5, -1.5) node[above, pos=.17]{\footnotesize Precoded CSI-RS};
				\draw[->, dashed] (-3.5, -1.5) -- (2.5, -1.5);
				\draw[->, dashed] (-3.5, -1.5) -- (-.5, -1.5);
				\draw[<-, dashed] (-3.5, -2) -- (5.5, -2) node[right]{$\mathbf{\hat{H}}$};
				\filldraw[black] (2.5, -2) circle (.07cm);
				\filldraw[black] (-.5, -2) circle (.07cm);
				\fill[pattern=north east lines, pattern color=EurecomBlue] (-3.5,-3) rectangle (5.5,-2) 
				node[pos=.5, fill=white]{\footnotesize Data communication phase};
				\filldraw[fill=white] (-5.5, -2.75) rectangle (-1.5, -2.25) node[pos=.5]{\footnotesize Set the mMIMO precoder};
				\draw[->, dashed] (-3.5, -3) node[left]{$\mathbf{VS}$} -- (5.5, -3);
				\draw[->, dashed] (-3.5, -3) -- (2.5, -3);
				\draw[->, dashed] (-3.5, -3) -- (-.5, -3);
				\draw[<-, dashed] (-3.5, -3.5) -- (5.5, -3.5) node[right]{$\mathbf{\hat{H}}$};
				\filldraw[black] (2.5, -3.5) circle (.07cm);
				\filldraw[black] (-.5, -3.5) circle (.07cm);
				\fill[pattern=north east lines, pattern color=EurecomBlue] (-3.5,-3.5) rectangle (5.5,-4.5) 
				node[pos=.5, fill=white]{\footnotesize Data communication phase};
				\filldraw[fill=white] (-5.5, -4.25) rectangle (-1.5, -3.75) node[pos=.5]{\footnotesize Set the mMIMO precoder};
				\draw[decorate,decoration={brace,amplitude=11pt}](6,-3) -- (6,-4.5) 
				node [black,midway,xshift=0.7cm, rotate=-90]{\footnotesize Channel Coherence Time};
				
				\draw[->, dashed] (-3.5, -4.5) node[left]{$\mathbf{VS}$} -- (5.5, -4.5);
				\draw[->, dashed] (-3.5, -4.5) -- (2.5, -4.5);
				\draw[->, dashed] (-3.5, -4.5) -- (-.5, -4.5);
				\draw[<-, dashed] (-3.5, -5) -- (5.5, -5) node[right]{$\mathbf{\hat{H}}$};
				\filldraw[black] (2.5, -5) circle (.07cm);
				\filldraw[black] (-.5, -5) circle (.07cm);
				\fill[pattern=north east lines, pattern color=EurecomBlue] (-3.5,-6) rectangle (5.5,-5) 
				node[pos=.5, fill=white]{\footnotesize Data communication phase};
				\filldraw[fill=white] (-5.5, -5.25) rectangle (-1.5, -5.75) node[pos=.5]{\footnotesize Set the mMIMO precoder};
				\draw[decorate,decoration={brace,amplitude=11pt, mirror}] (-5.77,5) -- (-5.77,-6.5) 
				node [black,midway,xshift=-0.7cm, rotate=90]{\footnotesize Beam Coherence Time};
			\end{tikzpicture}
		}
		\caption{Signaling sequence of the proposed coordinated beam selection \eqref{eq:coo_beam_sel} for $K = 3$. 
		The beam decision made at each UE leverages the \ac{D2D}-enabled long-term statistical information 
		(the \ac{PMI} $\mathcal{M}_k^{\text{BS}}\left(\mathbf{W}_k^* \right)$, for \eqref{eq:coo_beam_sel}) coming from the 
		higher-ranked UEs in a hierarchical fashion.}
		\label{fig:signaling}
	\end{figure}
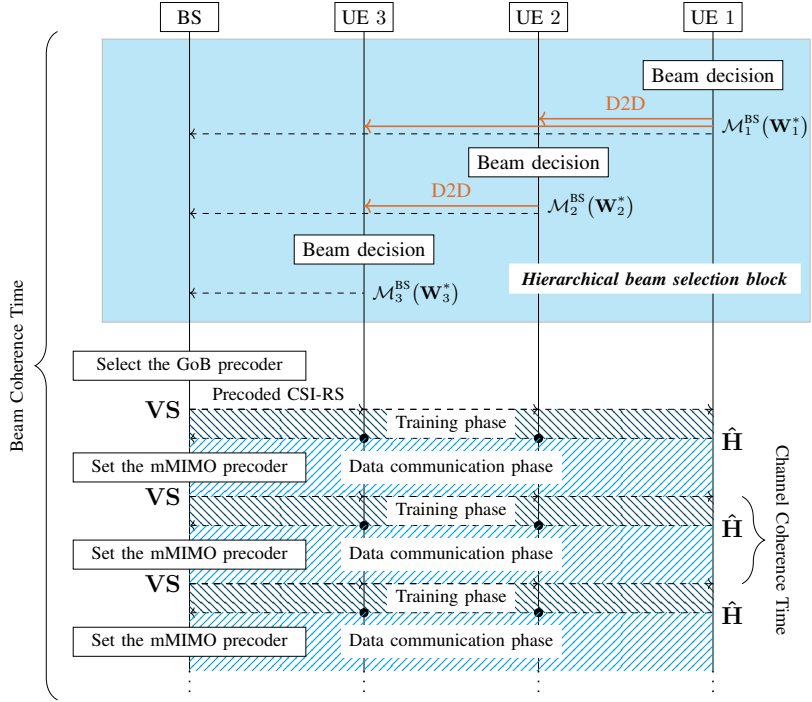
			
	Let us consider w.l.o.g. the beam selection at the $k$-th \ac{UE}, i.e. at the $k$-th step of the algorithm, for the algorithm \eqref{eq:csh_beam_sel}.
	The algorithms \eqref{eq:cov_beam_sel}-\eqref{eq:csh_beam_sel} can be regarded as a sub-case of \eqref{eq:csh_beam_sel}.
	We define the set $\mathcal{W}_{k-1} \triangleq \{ \mathbf{W}_1^*, \dots, \mathbf{W}_{k-1}^* \}$ containing the beam decisions which
	have been fixed prior to the $k$-th step. According to the hierarchical structure, the $k$-th \ac{UE} knows the set  $\mathcal{B}^{\text{fix}}
	(\mathcal{W}_{k-1}) \triangleq \cup_{j=1}^{k-1} \mathcal{M}_j^{\text{BS}}(\mathbf{W}_j^*)$ and the effective channel covariances
	$\bar{\mathbf{\Sigma}}_j, ~\forall j \in \{ 1, \dots, k-1 \}$. Therefore, the $k$-th \ac{UE} can \emph{i)} evaluate a \emph{partial} \ac{GCMD} 
	$\delta_k\left(\bar{\mathbf{\Sigma}}_1, \dots, \bar{\mathbf{\Sigma}}_k\right)$ and \emph{ii)} construct a \emph{partial} \ac{GoB} precoder 
	$\mathbf{V}_{k-1}$ containing the precoding vectors relative to the indexes in $\mathcal{B}^{\text{fix}}(\mathcal{W}_{k-1})$,
	i.e. $\text{col}\left(\mathbf{V}_{k-1}\right) = 
	\left\{ \mathbf{v}_v \in \mathcal{B}_{\text{BS}} : (v,w) \in \mathcal{B}^{\text{fix}}(\mathcal{W}_{k-1}) \right\}$.
	Likewise, the $k$-th \ac{UE} can compute a \emph{partial} $\omega(\mathcal{W}_{k-1})$.
					
	The proposed decentralized beam selection $\mathbf{W}_k^*$ at the $k$-th UE can be then expressed in a recursive manner as follows:
	\begin{equation} \label{eq:hier_coo_bs_ZF}
		\mathbf{W}_k^* = 
		\argmax_{\mathbf{W}_k} 
		~f_k \left( \big[ \mathbf{V}_k \mathbf{V}_{k-1} \big], \{ \mathbf{W}_k, \mathcal{W}_{k-1} \} \right),
	\end{equation}
	where $\text{col}(\mathbf{V}_k) = 
	\left\{ \mathbf{v}_v \in \mathcal{B}_{\text{BS}} : (v,w)  \in \mathcal{M}_k^{\text{BS}}(\mathbf{W}_k) \right\}$, and
	\begin{equation} \label{eq:obj_func_dec}
		f_k(\mathbf{V}, \mathcal{W}) \triangleq 
		\begin{cases}
			M_{\textnormal{UE}} \log_2 \left( 1 + \kappa M_{\textnormal{UE}}^{-1} \textnormal{Tr}\left(\bar{\mathbf{\Sigma}}_k\right) 
			\delta_k\left(\bar{\mathbf{\Sigma}}_1, \dots, \bar{\mathbf{\Sigma}}_k\right) \right) 
			& \eqref{eq:cov_beam_sel} \\
			\big(1-\omega(\mathcal{W}) \big) M_{\textnormal{UE}} \log_2 \left( 1 + \kappa M_{\textnormal{UE}}^{-1} 
			\textnormal{Tr}\left(\bar{\mathbf{\Sigma}}_k\right) \right) 
			& \eqref{eq:coo_beam_sel} \\
			\big(1-\omega(\mathcal{W}) \big) 
			M_{\textnormal{UE}} \log_2 \left( 1 + \kappa M_{\textnormal{UE}}^{-1} \textnormal{Tr}\left(\bar{\mathbf{\Sigma}}_k\right) 
			\delta_k\left(\bar{\mathbf{\Sigma}}_1, \dots, \bar{\mathbf{\Sigma}}_k\right) \right)  
			& \eqref{eq:csh_beam_sel}
		\end{cases}
	\end{equation}
	The intuition behind the proposed scheme is to let the $k$-th \ac{UE} select the $\mathbf{W}_k \in \mathcal{B}_{\text{UE}}$ 
	maximizing the $k$-th term of the sum in the respective objective function, in a \emph{greedy} manner. 
	The remaining constraint $(K-1)M_{\text{UE}} < M_{\text{BS}}$ can be enforced at the \ac{BS} through e.g. activating predefined beams. 
	
	\subsection{On Algorithm and Information Exchange Complexity}
		
		The proposed decentralized problem can be addressed using linear (exhaustive) search in the codebook $\mathcal{B}_{\text{UE}}$ at each \ac{UE}. 
		In particular, the linear search does not involve a large omputational burden as the $k$-th \ac{UE} has to evaluate the respective objective
		function in \eqref{eq:obj_func_dec} in $\binom{B_{\text{UE}}}{M_{\text{UE}}}$ points, which is \emph{generally small} in practical 
		implementations~\cite{Dahlman2018}.
		On the other hand, the direct solving of the optimization problems \eqref{eq:cov_beam_sel}-\eqref{eq:csh_beam_sel} requires 
		combinatorial (exhaustive) search. Note that the complexity of the proposed solution is independent from the number of antennas at the BS 
		and at the UE side. However, the complexity increases with the number of beams $M_{\text{UE}}$ at the UE side, which is somewhat dependent
		on $N_{\text{UE}}$ (to avoid bad spatial resolution)~\cite{Zirwas2016}.
		
		Table~\ref{tab:algorithms} summarizes the required information at each UE for the proposed beam selection policies based on 
		\eqref{eq:un_beam_sel}-\eqref{eq:csh_beam_sel}. The proposed scheme relies on acquiring and tracking channel second order statistics at each UE, 
		on a beam coherence time basis, as described in Fig.~\ref{fig:signaling}. In this respect, the beam coherence time must be long enough to avoid
		large side-link overhead, and to ensure the feasibility of the proposed scheme. The authors in~\cite{Va2017} perform a thorough investigation
		on the behavior of the beam coherence time for practical use cases. For example, the beam coherence time is found to be higher than $1$ second 
		under a vehicular scenario with \acp{UE} speed around 100 km/h communicating in the $60$ GHz band. Thus, it seems reasonable to assume that 
		the beam coherence time is long enough for the purposes of the proposed scheme.
		
\section{Simulation Results}
	
	\setlength{\belowcaptionskip}{-0.57em}
	We evaluate here the performance of the proposed decentralized beam selection algorithms. 
	We assume $N_{\text{BS}} = 64$ and $N_{\text{UE}} = 4$.
	The beamforming vectors in $\mathcal{B}_{\text{BS}}$ and $\mathcal{B}_{\text{UE}}$ are \ac{DFT}-based
	orthogonal beams, according to the codebook-based transmission in $3$GPP \ac{NR}~\cite{Dahlman2018}.
	Furthermore, we assume that the \acp{UE} are allowed to indicate at most $4$ relevant beam pairs each to
	the \ac{BS}, i.e. the \ac{PMI} $\mathcal{M}^{\text{BS}}_k(\mathbf{W}_k^*)$ is truncated to its $4$
	strongest elements $\forall k$. This is equivalent to the \emph{Type II} \ac{CSI} reporting in 
	\ac{NR}~\cite{Dahlman2018}. We assume that the \acp{UE} use the popular \ac{MMSE} method to estimate 
	their instantaneous effective channels (refer to \eqref{eq:mmse_estimate} and \eqref{eq:lem_lmmse}), 
	which are then fed back to the \ac{BS} for \ac{BD}-based precoder design (refer to 
	Fig.~\ref{fig:signaling}). The \emph{Zadoff-Chu} sequences are used for training the effective
	channel~\cite{Dahlman2018}. In our simulation, we consider transmission over $25$ adjacent resource blocks,
	each one consisting in $12$ sub-carriers and $14$ OFDM symbols~\cite{Dahlman2018}. 
	The total considered bandwidth is $5$ MHz,
	and the total number of available resource elements\footnote{In actual networks, control information
	occupies a portion of the available resource elements. We assume that such portion is negligible and
	that all resource elements are used for either \ac{CSI}-\ac{RS} or data transmission.} in the channel coherence time is 
	$T \triangleq \left(14 \cdot 12 \cdot 25\right) T_{\text{coh}}$, where $T_{\text{coh}}$ is expressed in ms.
	All the metrics in the next plots are averaged over $10000$ Monte-Carlo iterations with varying network scenarios.
	
	\subsection{Winner II Channel Model}
		
		The channel model used for the simulations is the cluster-based Winner II model, which extends 
		the $3$GPP spatial channel model. % The channel parameters are generated through statistical
		% distributions extracted from channel measurements. Several measurement campaigns provide the
		% background for the parametrization of the propagation scenarios. 
		In particular, we consider the urban micro-cell scenario operating at $2.1$ GHz. 
		In urban micro-cell scenarios, both the \ac{BS} and the \acp{UE} are assumed to be located outdoors 
		in an area where the streets are laid out in a Manhattan-like grid. 
		The \ac{UE} speed is set according to a uniform distribution in the interval $[5, 50]$ km/h.
		This scenario considers both line-of-sight and non-line-of-sight links. 
		Like in all cluster-based models, the channel realizations are generated through summing the
		contributions of the multiple paths within each cluster. Those paths come with their own small- and large-scale
		parameters such as amplitude, \ac{AoD} and \acs{AoA}. % The superposition of several paths results in
		% scorrelation between antenna elements and temporal fading with corresponding Doppler spectrum. 
		% Further information about the Winner II channel model can be found in~\cite{Kyosti2008}.
				
	\subsection{Results and Discussion}
				
		In what follows, we show and discuss the performance achieved via the proposed algorithms. 
		For benchmarking, we consider -- where appropriate -- the relative \ac{TDD} configuration with both
		\emph{perfect} and \emph{imperfect} \ac{CSI}\footnote{The curves relative to the TDD setting are obtained 
		after \ac{BD} precoding at the \ac{BS}, under the assumption of \emph{perfect channel reciprocity}. 
		The precoder is applied over the full-dimensional $K N_{\text{UE}} \times N_{\text{BS}}$ channel.
		In this respect, a clear advantage is experienced over the proposed \ac{FDD} solutions, which consider 
		the transmission over an \emph{effective} channel, as a result of the \ac{GoB} precoding.}. Moreover, 
		we will consider two different network scenarios: the \emph{i)} with \emph{randomly-located} \acp{UE}, 
		and the \emph{ii)} with \emph{closely-located} \acp{UE} \emph{(highly spatially-correlated} channels\emph{)}.
		
		We start with configuration \emph{i)}. 
		In Fig.~\ref{fig:sum_rate_vs_snr_K7}, we show the average effective network throughput as a function
		of the transmit \ac{SNR} for $K=7$ \acp{UE} and a channel coherence time $T_{\text{coh}} = 15$ ms. 
		As expected, higher throughput is achieved under TDD settings compared to the proposed FDD solutions~\cite{Flordelis2018}.
		In particular, the gain over FDD increases with the SNR.
		With respect to the proposed algorithms, both \eqref{eq:coo_beam_sel} and \eqref{eq:csh_beam_sel} 
		outperform the uncoordinated benchmark \eqref{eq:un_beam_sel}, with equal average effective throughput values 
		obtained with up to $10$ dBs less. Since $T_{\text{coh}}$ is small, the pre-log factor dominates the 
		log factor in \eqref{eq:ent}. Therefore, just a small performance gap divides \eqref{eq:coo_beam_sel} 
		from \eqref{eq:csh_beam_sel} and, as such, according to Table~\ref{tab:algorithms}, \eqref{eq:coo_beam_sel} 
		is preferable in this case (since much less information needs to be shared among the \acp{UE}). 
		On the other hand, the coordinated algorithm \eqref{eq:cov_beam_sel} performs even worse than the 
		uncoordinated benchmark \eqref{eq:un_beam_sel}. In particular, when the channel coherence time 
		$T_{\text{coh}}$ is small, shaping the covariances so as to maximize the \emph{spatial separability} 
		of the \acp{UE} is counter-effective. Some more insights on this are given in the next paragraph.
		Since the training overhead increases with $K$, the performance gain achieved by the coordinated
		algorithms \eqref{eq:coo_beam_sel} and \eqref{eq:csh_beam_sel} surges in 
		Fig.~\ref{fig:sum_rate_vs_snr_K11} for $K = 11$. For the same reason, the gap between
	 	\eqref{eq:cov_beam_sel} and all other solutions increases.
		
		\begin{figure}[h]
			\centering
			\begin{subfigure}[h]{0.485\textwidth}
				\centering
				\includegraphics[trim=0.17in 0.17in 0.17in 0.17in, clip, width=0.927\linewidth]{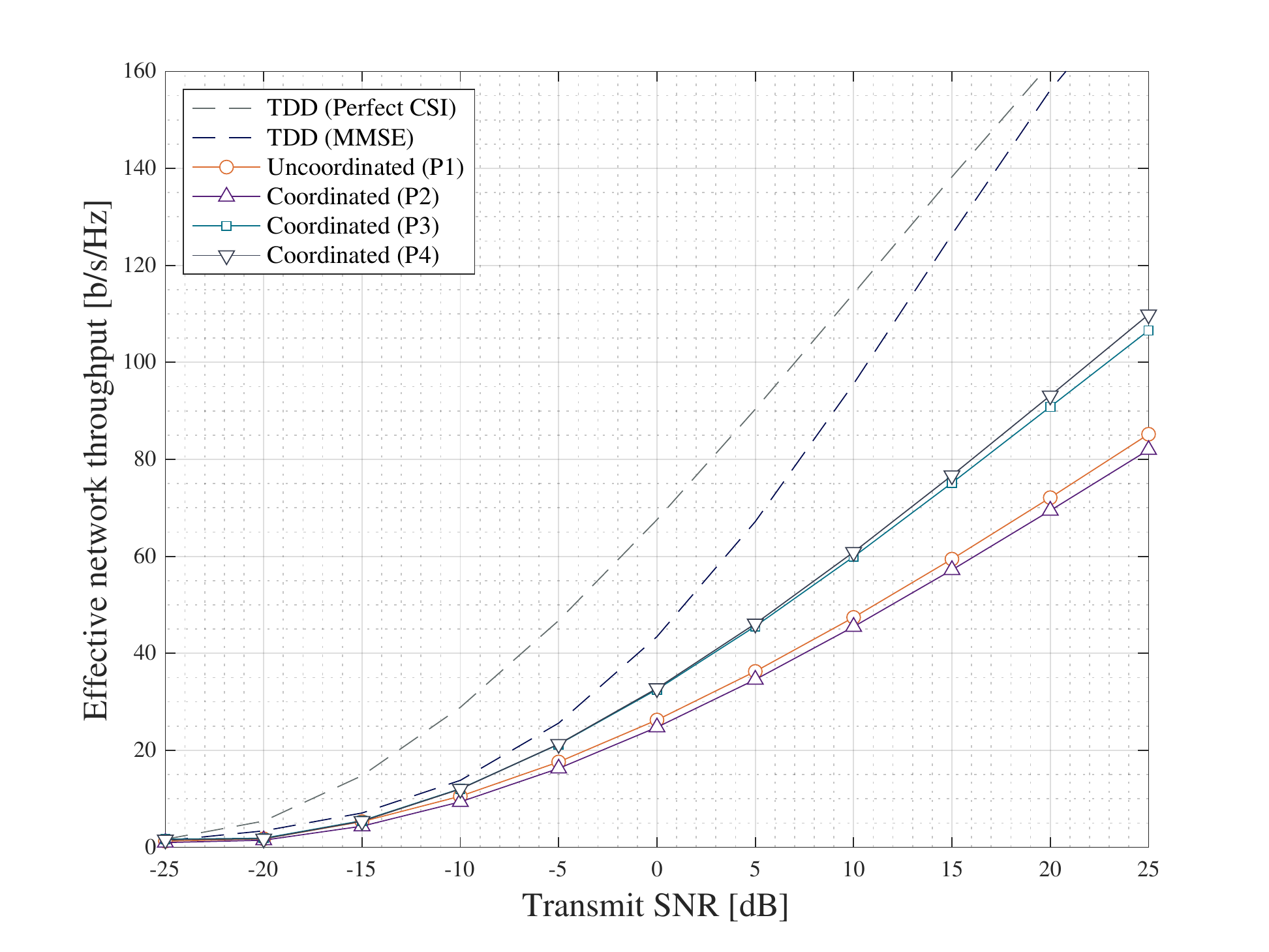}
				\caption{$K = 7$ \acp{UE}}
				\label{fig:sum_rate_vs_snr_K7}
			\end{subfigure}
			\begin{subfigure}[h]{0.485\textwidth}
				\centering
				\includegraphics[trim=0.17in 0.17in 0.17in 0.17in, clip, width=0.927\linewidth]{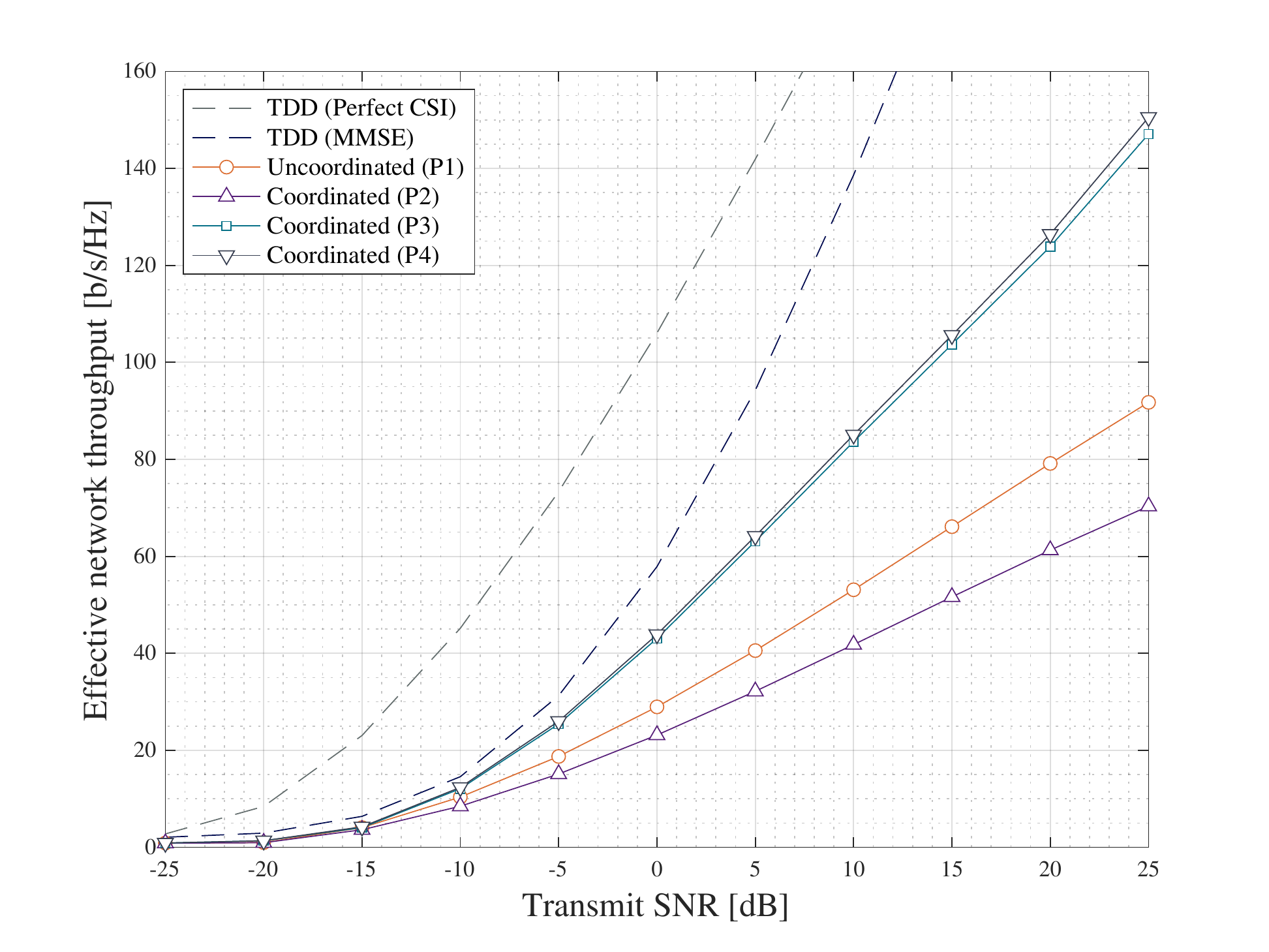}
				\caption{$K = 11$ \acp{UE}}
				\label{fig:sum_rate_vs_snr_K11}
			\end{subfigure}
			\caption{Average effective network throughput vs \ac{SNR} for (a) $K=7$ and (b) $K=11$ 
			\emph{randomly-located} \acp{UE}. 
			$M_{\text{UE}} = 3$ beams activated at each \ac{UE}. $T_{\text{coh}} = 15$ ms.
			The coordinated algorithms \eqref{eq:coo_beam_sel} and \eqref{eq:csh_beam_sel} outperform 
			the uncoordinated \eqref{eq:un_beam_sel}, as opposed to \eqref{eq:cov_beam_sel}.}
			\label{fig:sum_rate_vs_snr}
		\end{figure}
		
		Fig.~\ref{fig:gain_over_un_vs_cct_K7} shows the average throughput gain over the uncoordinated
		benchmark \eqref{eq:un_beam_sel} as a function of $T_{\text{coh}}$ for $K=7$ \acp{UE}. In particular,
		two areas can be identified: 
		\begin{subsubsection}{$T_{\textnormal{coh}} < 20$ \text{ms, i.e. vehicular or fast pedestrian
		channels}} where \eqref{eq:coo_beam_sel} and \eqref{eq:csh_beam_sel} have high gains compared to the
		other solutions (up to $45\%$) and where the coordinated algorithm \eqref{eq:cov_beam_sel} performs
		even worse than the uncoordinated \eqref{eq:un_beam_sel}.
		Indeed, as we can see in Fig.~\ref{fig:activ_beams_bs_vs_cct}, in order to achieve greater spatial
		separation across the \acp{UE}, the algorithm \eqref{eq:cov_beam_sel} activates a much greater number
		of beams at the \ac{BS} side, which is detrimental under \emph{fast-varying} channels.
		\end{subsubsection}
		\begin{subsubsection}{$T_{\textnormal{coh}} \ge 20$ ms, \text{i.e. pedestrian channels}}
			where the gap between \eqref{eq:coo_beam_sel}-\eqref{eq:csh_beam_sel} and \eqref{eq:un_beam_sel} reduces (up to $15\%$). 
			In particular, \eqref{eq:coo_beam_sel} converges to the uncoordinated benchmark \eqref{eq:un_beam_sel}. This is because the training
			overhead becomes negligible for long channel coherence times, and it is more important to focus on the log factor in \eqref{eq:ent}.
			For the same reason, \eqref{eq:cov_beam_sel} experiences gains over the uncoordinated solution \eqref{eq:un_beam_sel} for 
			$T_{\textnormal{coh}} \ge 20$ ms. The coordinated algorithm \eqref{eq:csh_beam_sel} converges to \eqref{eq:cov_beam_sel}. 
			Therefore, for long channel coherence times, \eqref{eq:cov_beam_sel} allows to avoid some additional coordination overhead, according to 
			Table~\ref{tab:algorithms} and, as such, is preferable.
		\end{subsubsection}
		
		The same reasoning holds for Fig.~\ref{fig:gain_over_un_vs_cct_K11} with $K=11$ \acp{UE}, where the positive and negative behaviors described
		above are intensified. In particular, for $T_{\text{coh}} < 20$ ms, \eqref{eq:coo_beam_sel} and \eqref{eq:csh_beam_sel} achieve up to $120 \%$ gain
		over \eqref{eq:un_beam_sel}, while $20\%$ loss is achieved with \eqref{eq:cov_beam_sel}.
		
		\begin{figure}[h]
			\centering
			\begin{subfigure}[h]{0.485\textwidth}
				\centering
				\includegraphics[trim=0.17in 0.17in 0.17in 0.17in, clip, width=0.927\linewidth]{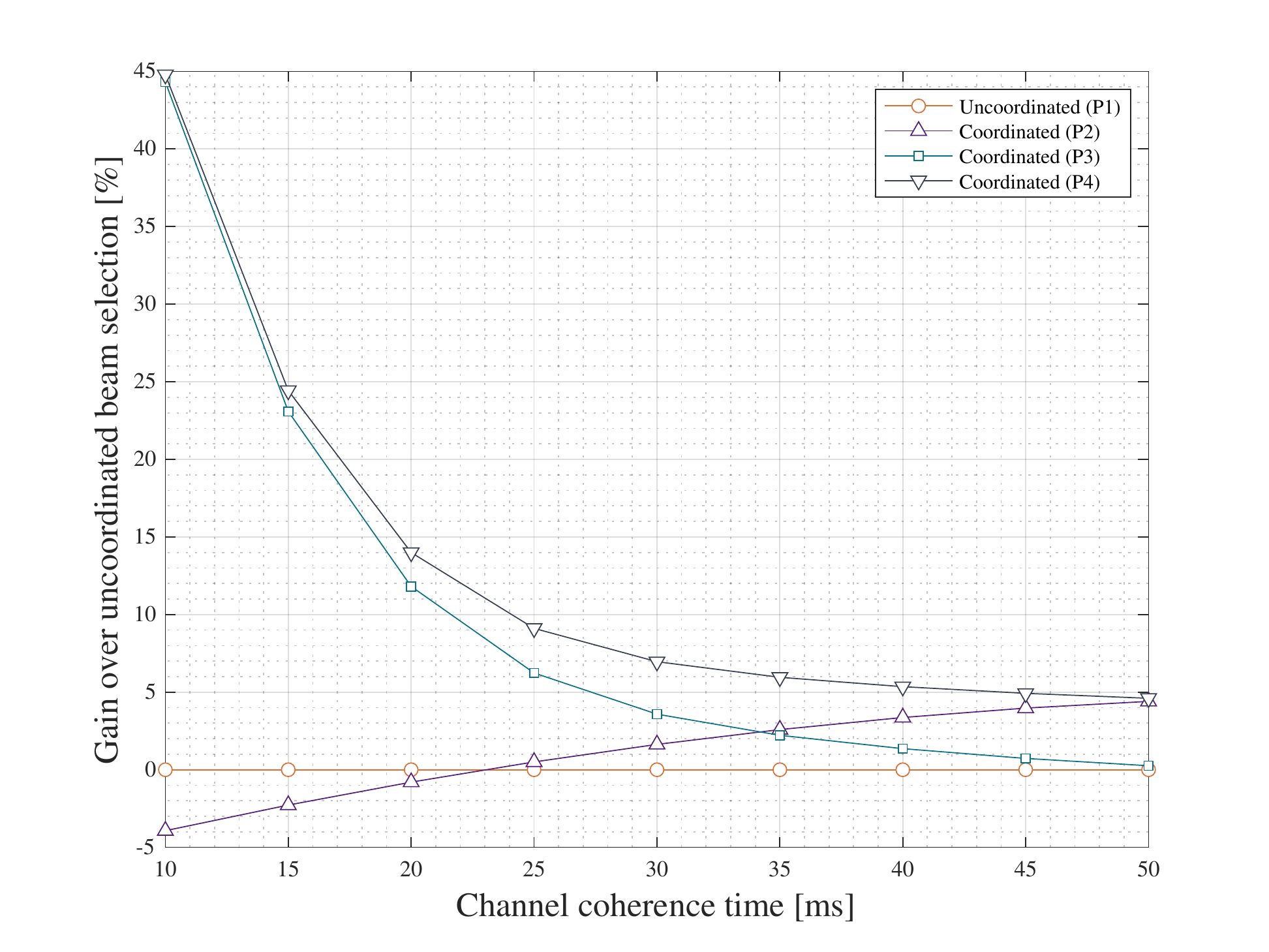}
				\caption{$K = 7$ \acp{UE}}
				\label{fig:gain_over_un_vs_cct_K7}
			\end{subfigure}
			\begin{subfigure}[h]{0.485\textwidth}
				\centering
				\includegraphics[trim=0.17in 0.17in 0.17in 0.17in, clip, width=0.927\linewidth]{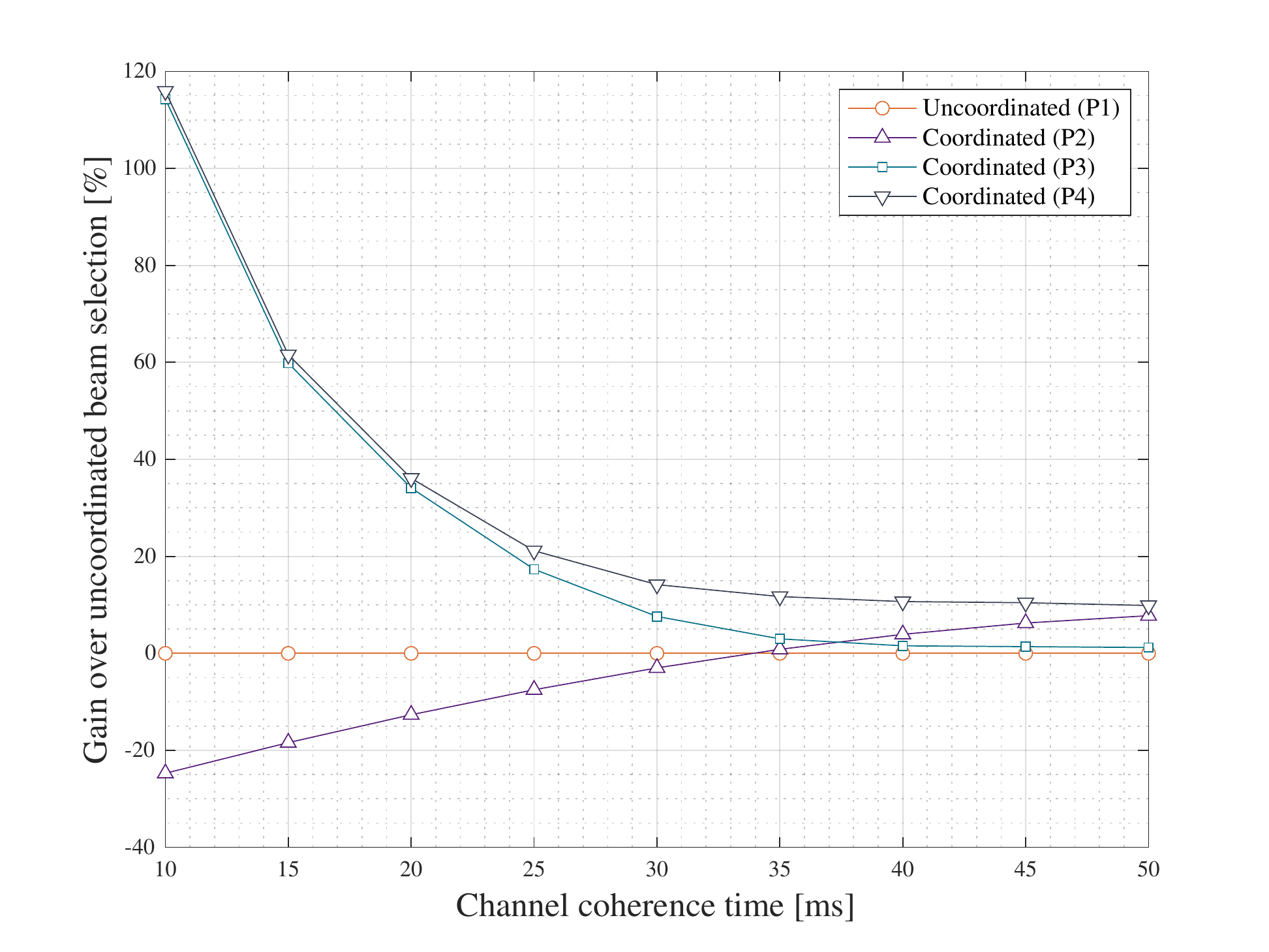}
				\caption{$K = 11$ \acp{UE}}
				\label{fig:gain_over_un_vs_cct_K11}
			\end{subfigure}
			\caption{Average effective throughput gain over uncoordinated beam selection \eqref{eq:un_beam_sel} vs $T_{\text{coh}}$ for $K = 7$ \acp{UE}. 
			The transmit \ac{SNR} is $11$ dB. Taking the pre-log factor into account is essential for an effective coordinated beam selection under 
			\emph{fast-varying} channels where $T_{\text{coh}} < 20$ ms.}
			\label{fig:gain_over_un_vs_cct}
		\end{figure}
		\setlength{\belowcaptionskip}{-1.57em}
		\begin{figure}[h]
			\centering
			\begin{subfigure}[h]{0.485\textwidth}
				\centering
				\includegraphics[trim=0.17in 0.17in 0.17in 0.17in, clip, width=0.927\linewidth]{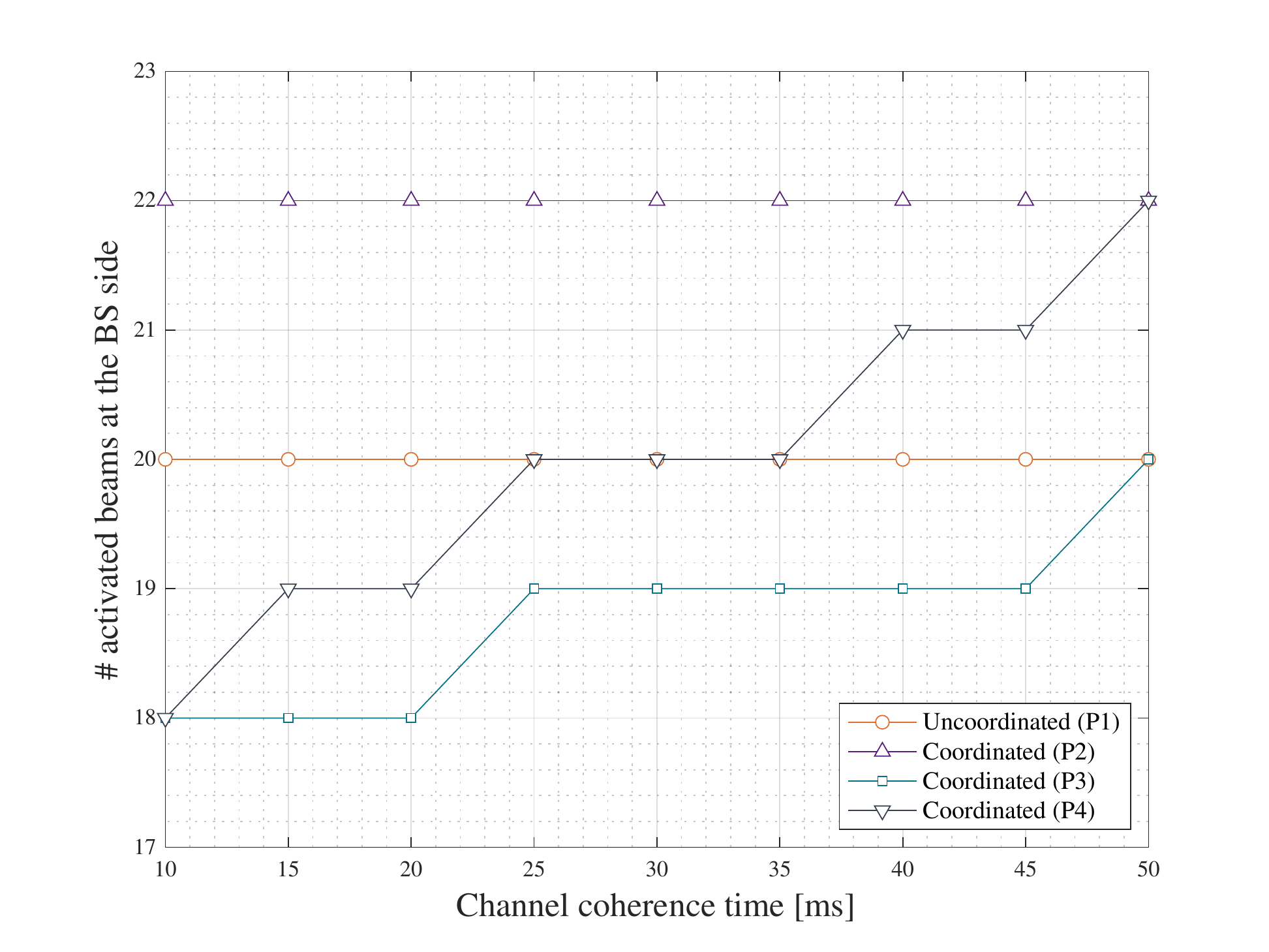}
				\caption{$K = 7$ \acp{UE}}
			\end{subfigure}
			\begin{subfigure}[h]{0.485\textwidth}
				\centering
				\includegraphics[trim=0.17in 0.17in 0.17in 0.17in, clip, width=0.927\linewidth]{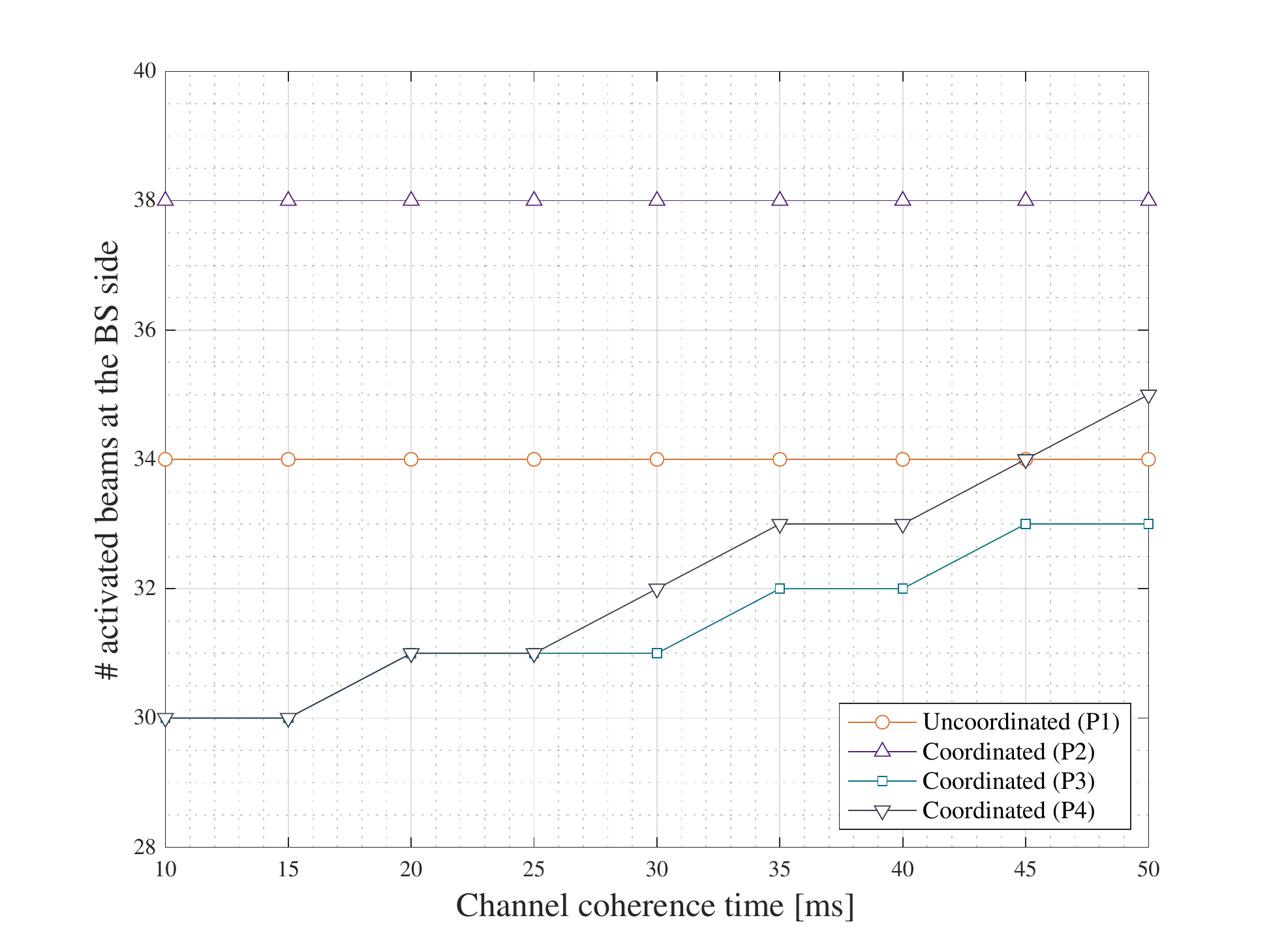}
				\caption{$K = 11$ \acp{UE}}
			\end{subfigure}
			\vspace{1em}
			\caption{Average $M_{\text{BS}}$ for the proposed algorithms vs $T_{\text{coh}}$ for (a) $K=7$ 
			and (b) $K=11$ \acp{UE}. The transmit \ac{SNR} is $11$ dB. The coordinated algorithm
			\eqref{eq:cov_beam_sel} activates more beams at the \ac{BS} side in order 
			to achieve greater spatial separation among the \acp{UE}.}
			\label{fig:activ_beams_bs_vs_cct}
		\end{figure}
		
		Let us now focus on the configuration \emph{ii)}, where higher spatial correlation is found among the \acp{UE}. 
		Fig.~\ref{fig:gain_over_un_vs_cct_K7_corr} shows the average throughput gain over the uncoordinated benchmark 
		as a function of the channel coherence time $T_{\text{coh}}$. We can see that now \eqref{eq:cov_beam_sel}
		outperforms \eqref{eq:un_beam_sel} for all the considered values of $T_{\text{coh}}$. 
		Indeed, due to the increasing spatial correlation among the \acp{UE}, the 
		multi-user interference becomes non-negligible even for small channel coherence times below $20$ ms.
		Moreover, in this case, the performance gain obtained through \eqref{eq:csh_beam_sel} justifies more
		the need to exchange some additional long-term information compared to the other solutions 
		(refer to Table~\ref{tab:algorithms}).
		
		\begin{figure}[h]
			\centering
			\includegraphics[trim=0.17in 0.17in 0.17in 0.17in, clip, width=0.457\columnwidth]
			{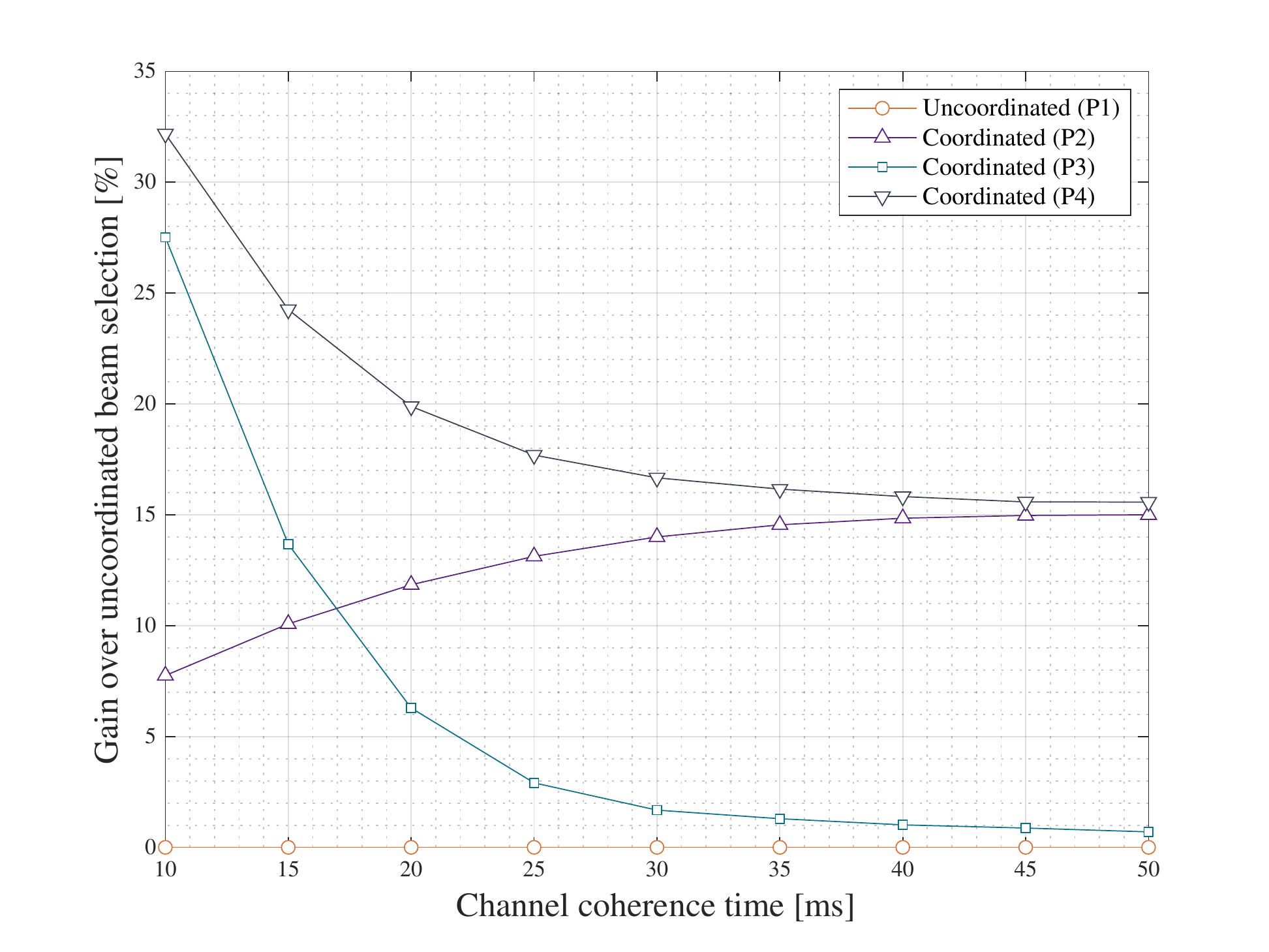}
			\caption{Average effective throughput gain over uncoordinated beam selection
			\eqref{eq:un_beam_sel} vs $T_{\text{coh}}$ for 
			$K = 7$ \emph{closely-located} \acp{UE}. The transmit \ac{SNR} is $11$ dB. 
			Owing to high spatial correlation among the \acp{UE}, \eqref{eq:cov_beam_sel}
			achieves high gains compared to the solutions which neglect the multi-user interference.}
			\label{fig:gain_over_un_vs_cct_K7_corr}
		\end{figure}
		
	\subsection{Effect of Feedback Quantization}
	
		In this section, we take into account the impact of limited feedback from the \acp{UE} to the \ac{BS} 
		on the overall \ac{DL} performance. In particular, we assume that the estimated channels are
		\emph{uniformly-quantized} with $q_\text{B}$ bits (element-wise quantization). 
		For a fairer comparison between \ac{FDD} and \ac{TDD} operation,
		we assume that, under \ac{TDD}, the quantization is applied as well on the (effective) channels which are fed 
		back to the \ac{UE} for coherent detection. Fig.~\ref{fig:ent_vs_qB} shows the effective network throughput as
		a function of the quantization bits $q_\text{B}$ for an \ac{SNR} equal to $11$ dB,
		and a $T_{\text{coh}} = 18$ ms. As expected, reducing $q_\text{B}$ entails a sharp throughput loss.
		However, it should be noted that such a loss starts when $q_{\text{B}} \le 7$, except under \ac{TDD} settings.
		With respect to the proposed algorithms, adopting a small $q_\text{B}$ reduces the performance gaps as well.
		At the same time, quantizing with a smaller $q_\text{B}$ reduces the feedback overhead in the \ac{UL} band.
		In this respect, feedback overhead and coordination performance must be balanced.
		% Note that element-wise quantization is assumed to derive Fig.~\ref{fig:ent_vs_qB}.
		\vspace{-0.37cm}
		\begin{figure}[h]
			\centering
			\begin{subfigure}[h]{0.485\textwidth}
				\centering
				\includegraphics[trim=0.17in 0.17in 0.17in 0.17in, clip, width=0.927\linewidth]{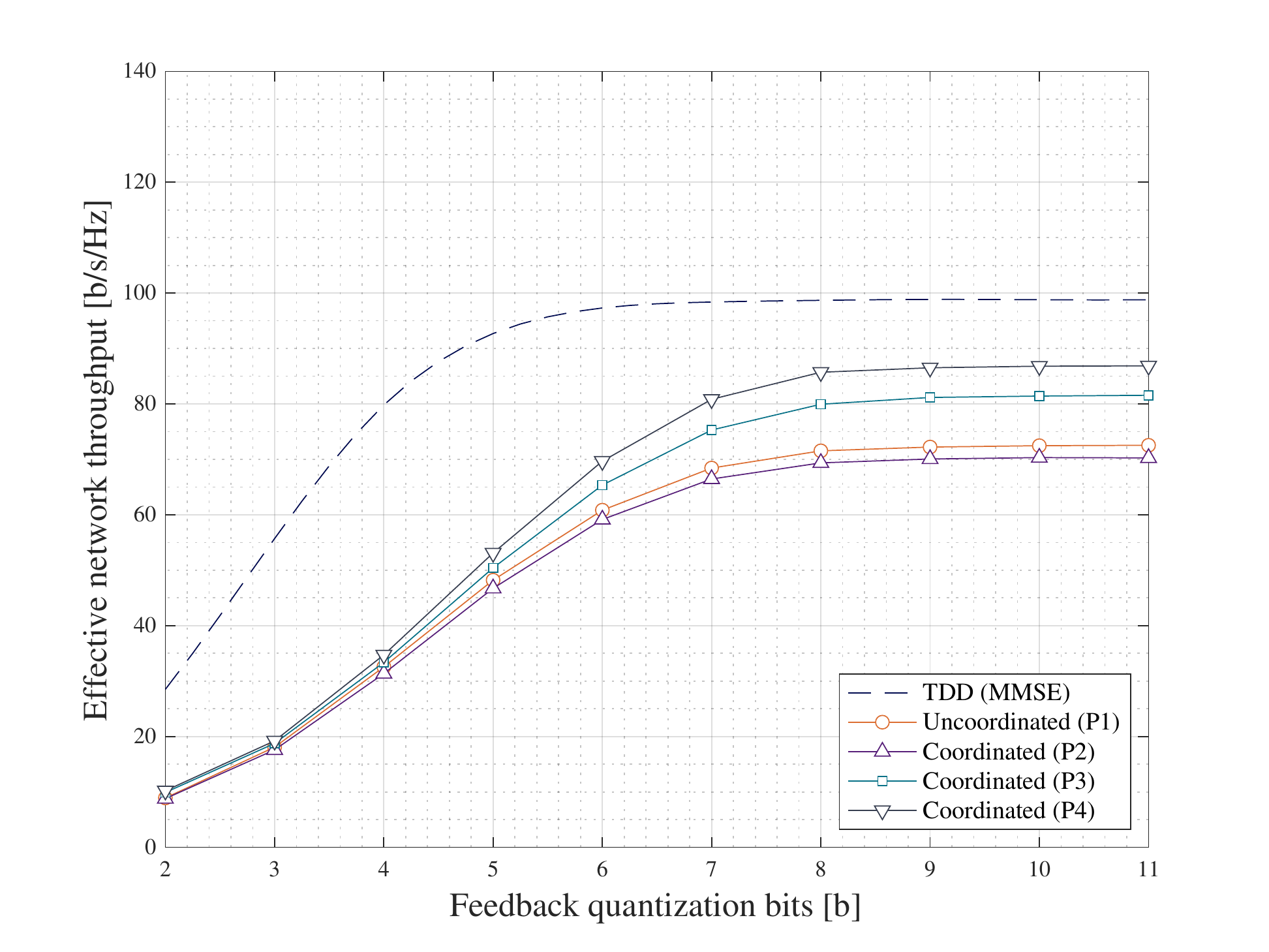}
				\caption{$K = 7$ \acp{UE}}
				\label{fig:ent_vs_qB_K7}
			\end{subfigure}
			\begin{subfigure}[h]{0.485\textwidth}
				\centering
				\includegraphics[trim=0.17in 0.17in 0.17in 0.17in, clip, width=0.927\linewidth]{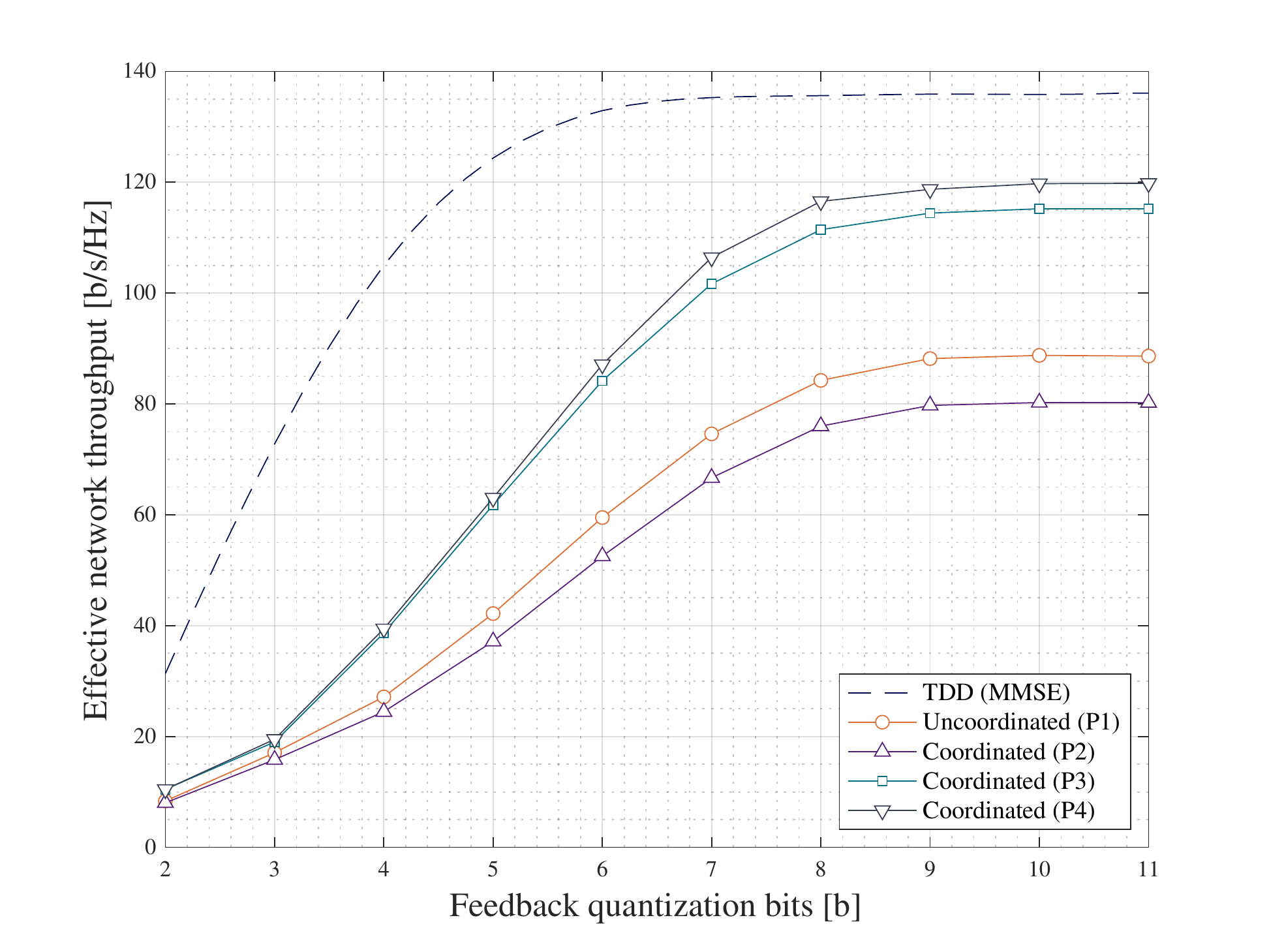}
				\caption{$K = 11$ \acp{UE}}
				\label{fig:ent_vs_qB_K11}
			\end{subfigure}
			\vspace{1em}
			\caption{Average effective network throughput vs $q_\text{B}$ for $K = 7$ \acp{UE}. The transmit \ac{SNR} is $11$ dB.
			$T_{\text{coh}} = 18$ ms. Reducing $q_\text{B}$ reduces the feedback overhead, but decreases performance \emph{sharply}.}
			\label{fig:ent_vs_qB}
		\end{figure}
	
\section{Conclusions}
	
	In this paper, we have shown that beam-domain coordination between the \acp{UE} offers a convenient 
	means to improve the performance in \ac{FDD} \ac{mMIMO} networks under the \ac{GoB} design assumption. 
	We have proposed a decentralized beam selection algorithm exploiting long-term statistical information and 
	its exchange through D$2$D side-links. The proposed scheme explores the interesting \emph{trade-off} between 
	\emph{i)} harvesting large channel gain, \emph{ii)} avoiding multi-user interference (low spatial separation among the \acp{UE}), 
	and \emph{iii)} minimizing the training overhead, which arises in this context. 
	Simulation results demonstrate the effectiveness of the proposed algorithm. 
	In particular, under fast pedestrian or vehicular channels where the channel coherence time is below $20$ ms, 
	the proposed scheme gives consistently better performance in terms of overall network throughput than uncoordinated beam selection 
	or schemes that do not properly address the above trade-off.

\bibliography{Bibl}

\begin{thebibliography}{10}
\providecommand{\url}[1]{#1}
\csname url@samestyle\endcsname
\providecommand{\newblock}{\relax}
\providecommand{\bibinfo}[2]{#2}
\providecommand{\BIBentrySTDinterwordspacing}{\spaceskip=0pt\relax}
\providecommand{\BIBentryALTinterwordstretchfactor}{4}
\providecommand{\BIBentryALTinterwordspacing}{\spaceskip=\fontdimen2\font plus
\BIBentryALTinterwordstretchfactor\fontdimen3\font minus
  \fontdimen4\font\relax}
\providecommand{\BIBforeignlanguage}[2]{{%
\expandafter\ifx\csname l@#1\endcsname\relax
\typeout{** WARNING: IEEEtran.bst: No hyphenation pattern has been}%
\typeout{** loaded for the language `#1'. Using the pattern for}%
\typeout{** the default language instead.}%
\else
\language=\csname l@#1\endcsname
\fi
#2}}
\providecommand{\BIBdecl}{\relax}
\BIBdecl

\bibitem{Maschietti2019c}
F.~Maschietti, G.~Fodor, D.~Gesbert, and P.~de~Kerret, ``Coordinated beam
  selection for training overhead reduction in {FDD} massive {MIMO},''
  \emph{Proc. IEEE ISWCS}, Aug. 2019.

\bibitem{Larsson2014}
E.~G. Larsson, O.~Edfors, F.~Tufvesson, and T.~L. Marzetta, ``Massive {MIMO}
  for next generation wireless systems,'' \emph{IEEE Commun. Mag.}, Feb. 2014.

\bibitem{Choi2014}
J.~{Choi}, D.~J. {Love}, and P.~{Bidigare}, ``Downlink training techniques for
  {FDD} massive {MIMO} systems: Open-loop and closed-loop training with
  memory,'' \emph{IEEE J. Sel. Topics Signal Process.}, Oct. 2014.

\bibitem{Yin2013}
H.~Yin, D.~Gesbert, M.~Filippou, and Y.~Liu, ``A coordinated approach to
  channel estimation in large-scale multiple-antenna systems,'' \emph{IEEE J.
  Sel. Areas Commun.}, Feb. 2013.

\bibitem{Adhikary2013}
A.~Adhikary, J.~Nam, J.~Ahn, and G.~Caire, ``Joint spatial division and
  multiplexing --- {The} large-scale array regime,'' \emph{IEEE Trans. Inf.
  Theory}, Oct. 2013.

\bibitem{Gao2015}
X.~{Gao}, O.~{Edfors}, F.~{Rusek}, and F.~{Tufvesson}, ``Massive {MIMO}
  performance evaluation based on measured propagation data,'' \emph{IEEE
  Trans. Wireless Commun.}, July 2015.

\bibitem{Khalilsarai2019}
M.~B. Khalilsarai, S.~Haghighatshoar, X.~Yi, and G.~Caire, ``{FDD} massive
  {MIMO} via {UL/DL} channel covariance extrapolation and active channel
  sparsification,'' \emph{IEEE Trans. Wireless Commun.}, Jan. 2019.

\bibitem{Newinger2015}
M.~{Newinger} and W.~{Utschick}, ``Covariance shaping for interference
  coordination in cellular wireless communication systems,'' in \emph{Proc.
  IEEE ASILOMAR}, Nov. 2015.

\bibitem{Moghadam2017}
N.~N. {Moghadam}, H.~{Shokri-Ghadikolaei}, G.~{Fodor}, M.~{Bengtsson}, and
  C.~{Fischione}, ``Pilot precoding and combining in multiuser {MIMO}
  networks,'' \emph{IEEE J. Sel. Areas Commun.}, July 2017.

\bibitem{Mursia2018}
P.~{Mursia}, I.~{Atzeni}, D.~{Gesbert}, and L.~{Cottatellucci}, ``Covariance
  shaping for massive {MIMO} systems,'' in \emph{Proc. IEEE GLOBECOM}, Dec.
  2018.

\bibitem{Bajwa2010}
W.~U. Bajwa, J.~Haupt, A.~M. Sayeed, and R.~Nowak, ``Compressed channel
  sensing: A new approach to estimating sparse multipath channels,''
  \emph{Proc. IEEE}, June 2010.

\bibitem{Rao2014}
X.~Rao and V.~K.~N. Lau, ``Distributed compressive {CSIT} estimation and
  feedback for {FDD} multi-user massive {MIMO} systems,'' \emph{IEEE Trans.
  Signal Process.}, June 2014.

\bibitem{Gao2016}
Z.~{Gao}, L.~{Dai}, W.~{Dai}, B.~{Shim}, and Z.~{Wang}, ``Structured
  compressive sensing-based spatio-temporal joint channel estimation for {FDD}
  massive {MIMO},'' \emph{IEEE Trans. Commun.}, Feb. 2016.

\bibitem{Shen2016}
J.~{Shen}, J.~{Zhang}, E.~{Alsusa}, and K.~B. {Letaief}, ``Compressed {CSI}
  acquisition in {FDD} massive {MIMO}: How much training is needed?''
  \emph{IEEE Trans. Wireless Commun.}, June 2016.

\bibitem{Dai2018}
J.~{Dai}, A.~{Liu}, and V.~K.~N. {Lau}, ``{FDD} massive {MIMO} channel
  estimation with arbitrary {2D}-array geometry,'' \emph{IEEE Trans. Signal
  Process.}, May 2018.

\bibitem{Martinez2016}
A.~O. {Mart\`inez}, E.~{De Carvalho}, and J.~{\O{}}. {Nielsen}, ``Massive
  {MIMO} properties based on measured channels: Channel hardening, user
  decorrelation and channel sparsity,'' in \emph{Proc. IEEE ASILOMAR}, Nov.
  2016.

\bibitem{Ding2018}
Y.~{Ding} and B.~D. {Rao}, ``Dictionary learning-based sparse channel
  representation and estimation for {FDD} massive {MIMO} systems,'' \emph{IEEE
  Trans. Wireless Commun.}, Aug. 2018.

\bibitem{Luo2017}
X.~{Luo}, P.~{Cai}, X.~{Zhang}, D.~{Hu}, and C.~{Shen}, ``A scalable framework
  for {CSI} feedback in {FDD} massive {MIMO} via {DL} path aligning,''
  \emph{IEEE Trans. Signal Process.}, Sep. 2017.

\bibitem{Zhang2018}
X.~{Zhang}, L.~{Zhong}, and A.~{Sabharwal}, ``Directional training for {FDD}
  massive {MIMO},'' \emph{IEEE Trans. Wireless Commun.}, Aug. 2018.

\bibitem{Shen2018}
W.~{Shen}, L.~{Dai}, B.~{Shim}, Z.~{Wang}, and R.~W. {Heath}, ``Channel
  feedback based on {AoD}-adaptive subspace codebook in {FDD} massive {MIMO}
  systems,'' \emph{IEEE Trans. Commun.}, Nov. 2018.

\bibitem{Rottenberg2020}
F.~{Rottenberg}, T.~{Choi}, P.~{Luo}, C.~J. {Zhang}, and A.~F. {Molisch},
  ``Performance analysis of channel extrapolation in {FDD} massive {MIMO}
  systems,'' \emph{IEEE Trans. Wireless Commun.}, 2020.

\bibitem{Vasisht2016}
D.~{Vasisht}, S.~{Kumar}, H.~{Rahul}, and D.~{Katabi}, ``Eliminating channel
  feedback in next-generation cellular networks,'' in \emph{Proc. ACM SIGCOMM},
  2016.

\bibitem{Yang2018}
W.~{Yang}, L.~{Chen}, and Y.~{Liu}, ``Super-resolution for achieving frequency
  division duplex ({FDD}) channel reciprocity,'' in \emph{Proc. IEEE SPAWC},
  2018.

\bibitem{Arnold2019}
M.~Arnold, S.~D{\"{o}}rner, S.~Cammerer, S.~Yan, J.~Hoydis, and S.~ten Brink,
  ``Towards practical {FDD} massive {MIMO}: {CSI} extrapolation driven by deep
  learning and actual channel measurements,'' in \emph{Proc. IEEE ASILOMAR},
  2019.

\bibitem{Yang2019}
Y.~{Yang}, F.~{Gao}, G.~Y. {Li}, and M.~{Jian}, ``Deep learning-based downlink
  channel prediction for {FDD} massive {MIMO} system,'' \emph{IEEE Commun.
  Lett.}, 2019.

\bibitem{Choi2020}
H.~{Choi} and J.~{Choi}, ``Downlink extrapolation for {FDD} multiple antenna
  systems through neural network using extracted uplink path gains,''
  \emph{IEEE Access}, 2020.

\bibitem{3GPP2018c}
{3GPP}, ``{NR}; physical layer procedures for data - {Rel. 15},'' {TS 38.214}.
  Dec. 2018.

\bibitem{Dahlman2018}
E.~Dahlman, S.~Parkvall, and J.~Sk\"old, \emph{{5G NR}: the Next Generation
  Wireless Access Technology}.\hskip 1em plus 0.5em minus 0.4em\relax Academic
  Press, 2018.

\bibitem{Kim2013}
C.~{Kim}, T.~{Kim}, and J.~{Seol}, ``Multi-beam transmission diversity with
  hybrid beamforming for {MIMO-OFDM} systems,'' in \emph{Proc. IEEE Globecom
  Workshops}, Dec. 2013.

\bibitem{Flordelis2018}
J.~Flordelis, F.~Rusek, F.~Tufvesson, E.~G. Larsson, and O.~Edfors, ``Massive
  {MIMO} performance --- {TDD} versus {FDD}: What do measurements say?''
  \emph{IEEE Trans. Wireless Commun.}, Apr. 2018.

\bibitem{Zirwas2016}
W.~Zirwas, M.~B. Amin, and M.~Sternad, ``Coded {CSI} reference signals for {5G}
  --- {Exploiting} sparsity of {FDD} massive {MIMO} radio channels,'' in
  \emph{Proc. IEEE WSA}, Mar. 2016.

\bibitem{Xiong2017}
X.~{Xiong}, X.~{Wang}, X.~{Gao}, and X.~{You}, ``Beam-domain channel estimation
  for {FDD} massive {MIMO} systems with optimal thresholds,'' \emph{IEEE Trans.
  Wireless Commun.}, July 2017.

\bibitem{Fodor2016}
G.~Fodor, S.~Roger, N.~Rajatheva, S.~B. Slimane, T.~Svensson, P.~Popovski,
  J.~M. B.~D. Silva, and S.~Ali, ``An overview of device-to-device
  communications technology components in {METIS},'' \emph{IEEE Access}, June
  2016.

\bibitem{3GPP2018d}
{3GPP}, ``{NR}; study on {NR} {V2X},'' {Work item RP-181429}. June 2018.

\bibitem{Kay1993}
S.~M. Kay, \emph{Fundamentals Of Statistical Signal Processing: Estimation
  Theory}.\hskip 1em plus 0.5em minus 0.4em\relax Prentice Hall, 1993.

\bibitem{Ayach2014}
O.~E. Ayach, S.~Rajagopal, S.~Abu-Surra, Z.~Pi, and R.~W. Heath, ``Spatially
  sparse precoding in millimeter wave {MIMO} systems,'' \emph{IEEE Trans.
  Wireless Commun.}, Mar. 2014.

\bibitem{Alkhateeb2015}
A.~Alkhateeb, G.~Leus, and R.~W. Heath, ``Limited feedback hybrid precoding for
  multi-user millimeter wave systems,'' \emph{IEEE Trans. Wireless Commun.},
  Nov. 2015.

\bibitem{Spencer2004}
Q.~H. {Spencer}, A.~L. {Swindlehurst}, and M.~{Haardt}, ``Zero-forcing methods
  for downlink spatial multiplexing in multiuser {MIMO} channels,'' \emph{IEEE
  Trans. Signal Process.}, Feb. 2004.

\bibitem{Heath2016}
R.~W. Heath, N.~Gonz\'alez-Prelcic, S.~Rangan, W.~Roh, and A.~M. Sayeed, ``An
  overview of signal processing techniques for mmwave {MIMO} systems,''
  \emph{IEEE J. Sel. Topics Signal Process.}, Apr. 2016.

\bibitem{Va2017}
V.~Va, J.~Choi, and R.~W. Heath, ``The impact of beamwidth on temporal channel
  variation in vehicular channels and its implications,'' \emph{IEEE Trans.
  Veh. Technol.}, June 2017.

\bibitem{Tse2005}
D.~Tse and P.~Viswanath, \emph{Fundamentals of Wireless Communication}.\hskip
  1em plus 0.5em minus 0.4em\relax Cambridge University Press, 2005.

\bibitem{Bjornson2018}
E.~{Bj\"ornson}, J.~{Hoydis}, and L.~{Sanguinetti}, ``Massive {MIMO} has
  unlimited capacity,'' \emph{IEEE Trans. Wire. Commun.}, Jan. 2018.

\end{thebibliography}
\bibliographystyle{IEEEtran}
				
\end{document}